\documentclass{article}

\usepackage{graphicx}
\usepackage{amsmath}
\usepackage{authblk}
\usepackage{amsthm}
\usepackage{float}
\usepackage{amsfonts}
\usepackage[superscript,biblabel]{cite}
\usepackage[margin=0.5in]{geometry}

\newtheorem{theorem}{Theorem}

\begin{document}

\title{Analysis of stepped wedge cluster randomized trials in the presence of a time-varying treatment effect}

\author[1]{Avi Kenny\thanks{Corresponding author. Department of Biostatistics, University of Washington, 3980 15th Ave NE Box 351617, Seattle, WA 98195. avikenny@uw.edu}}
\author[1]{Emily Voldal}
\author[1]{Fan Xia}
\author[1]{Patrick J. Heagerty}
\author[1]{James P. Hughes}

\affil[1]{Department of Biostatistics, University of Washington, Seattle, WA, USA}

\maketitle

\abstract{Stepped wedge cluster randomized controlled trials are typically analyzed using models that assume the full effect of the treatment is achieved instantaneously. We provide an analytical framework for scenarios in which the treatment effect varies as a function of exposure time (time since the start of treatment) and define the ``effect curve'' as the magnitude of the treatment effect on the linear predictor scale as a function of exposure time. The ``time-averaged treatment effect'', (TATE) and ``long-term treatment effect'' (LTE) are summaries of this curve. We analytically derive the expectation of the estimator resulting from a model that assumes an immediate treatment effect and show that it can be expressed as a weighted sum of the time-specific treatment effects corresponding to the observed exposure times. Surprisingly, although the weights sum to one, some of the weights can be negative. This implies that the estimator may be severely misleading and can even converge to a value of the opposite sign of the true TATE or LTE. We describe several models that can be used to simultaneously estimate the entire effect curve, the TATE, and the LTE, some of which make assumptions about the shape of the effect curve. We evaluate these models in a simulation study to examine the operating characteristics of the resulting estimators and apply them to two real datasets.}

\section{Introduction}\label{sec_intro}

Cluster randomized trials (CRTs) involve randomizing groups of individuals to a treatment or control condition, and are often conducted when individual randomization is impractical. One type of CRT design is the stepped wedge, which has seen increased usage in recent years.\cite{hughes2015current} In a stepped wedge CRT, all clusters begin in the control state and eventually switch over to the treatment state in a staggered manner, with a random assignment of clusters to crossover times, or ``sequences''. Cluster randomized designs are often adopted to evaluate the impact of health care system level interventions through modification of the systematic processes used to treat patients. Data are typically collected from all clusters at each time point, often through a series of cross-sectional surveys. Some unique aspects of this design include the partial confounding of the treatment effect with time and the fact that all clusters are observed in both the control state and the treatment state. The strengths and limitations of the stepped wedge design have been extensively discussed in recent years.\cite{mdege2011systematic,kotz2012use,mdege2012there,kotz2012researchers,hemming2013stepped,hemming2015stepped,kotz2013stepped} In particular, the stepped wedge design is useful for situations in which it is considered unethical to entirely withhold the treatment from a portion of participants due to lack of equipoise and/or if it is not possible to implement the treatment simultaneously to all participants for logistical, financial, or other reasons.\cite{brown2006stepped}\\

Standard statistical models for analyzing data from stepped wedge CRTs typically include a treatment indicator variable, indexed by cluster and time, that equals zero when the cluster is in the control state and one when the cluster is in the treatment state.\cite{hussey2007design} The coefficient of this indicator variable can then interpreted as the treatment effect. This modeling choice implicitly assumes that the full effect of the treatment is reached immediately (i.e. within a single time step) and does not increase or decrease thereafter; for this reason, we refer to such models as \textit{immediate treatment} (IT) models. Thus, the IT model assumes that the shape of the \textit{effect curve} -- the level of the treatment effect as a function of time since the start of the treatment -- is fully known; this shape is depicted in Figure \ref{effect_curves}a. However, with the stepped wedge design, once each cluster crosses over to the treatment condition it will experience multiple time periods under the treatment condition with the possibility that the magnitude of treatment effect will change with increasing treatment experience. Such effect modification with increasing time is often referred to as the ``learning curve'' and is well-documented in areas such as surgery.\cite{hopper2007learning} As such, the assumption of an immediate treatment effect may be violated in some settings.\\

We follow Nickless et al.\cite{nickless2018mixed} and use the term ``exposure time'' to refer to the amount of time that has passed since the start of the treatment for a given cluster, where the start of the treatment corresponds to exposure time 0. This is contrasted with ``study time'', which is the amount of time that has passed since the start of the study. There are several ways in which a treatment effect can vary with exposure time. The effect might not be realized until the second or third time point following the treatment, but then reaches its full effect almost immediately; we refer to this as a ``delayed effect'' (Figure \ref{effect_curves}b). Alternatively, the magnitude of the effect might vary as a function of exposure time (Figure \ref{effect_curves}c,d,e). It is also possible for the treatment effect to vary as a function of study time (e.g. an intervention is more effective in the summer versus the winter), but we do not consider this case in this paper.\\

\begin{figure}[H]
\centerline{\includegraphics[width=6in]{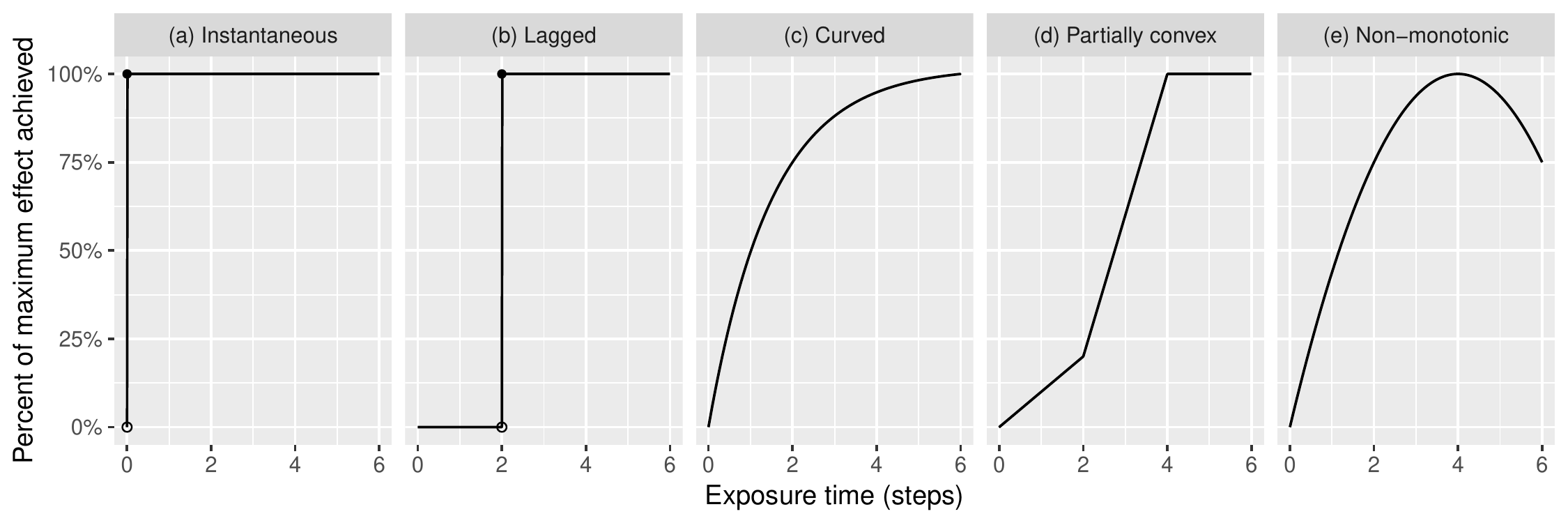}}
\caption{Several possible effect curves: treatment effect as a function of exposure time\label{effect_curves}}
\end{figure}

In some scenarios, the assumption of an immediate treatment effect seems justifiable. For example, it is reasonable to think that the effect of the implementation of a surgical safety checklist on patient outcomes would be immediate.\cite{haugen2015effect} However, this assumption feels less plausible in many other scenarios. For example, the effect of starting a new exercise regimen on clinical depression may take weeks or months to achieve, with the level of the effect increasing with time.\cite{craft1998effect} As another example, community health programs often involve a series of training modules for health workers that occur over a period of months, each of which is designed to have a separate effect on under-five mortality.\cite{white2018community}\\

The problem of a time-varying treatment effect in the context of stepped wedge CRTs was first considered by Granston et al.\cite{granston2014addressing}, who demonstrated that modeling the treatment effect using an indicator function can lead to biased estimates and invalid inference when the true effect curve is time-varying. They introduced a parametric model to account for the time-varying effect which assumes that the effect curve falls in a two-parameter family of concave functions. However, this model involves a complex two-stage estimation process and leads to incorrect inference when the true effect curve does not fall within this family, and fails to converge when the true effect curve is at least partially convex, as in Figure ~\ref{effect_curves}d (see supplementary material).\\

Hughes et al.\cite{hughes2015current} further considered the problem of a time-varying treatment effect and noted that in certain cases, this issue can be prevented by careful study design, such as increasing the length of time between steps, switching from an outcome endpoint that may take years to change to a process endpoint that will change more rapidly, or including a ``washout period'' immediately after implementation during which no data are collected. They also suggested the approach of changing the indicator variable representing the presence of the treatment such that this variable can take on fixed values between zero and one. However, this assumes that the shape of the effect curve is fully known.\\

Hemming et al.\cite{hemming2017analysis} described a model that includes fixed effects corresponding to study time by treatment interaction terms, with time treated as a categorical variable. They concluded that the confidence intervals for each interaction term coefficient were too wide for the model to be useful; however, they did not consider combining the interaction term coefficients into a single estimator.\\

As part of a simulation study, Nickless et al.\cite{nickless2018mixed} tested several models that account for time-varying treatment effects. They consider four types of interaction terms between time and treatment, which differ in terms of how time is modeled. Time can be either study time or exposure time, and can be modeled as continuous or categorical. The model that includes an interaction between treatment and categorical study time is equivalent to the Hemming et al.\cite{hemming2017analysis} model.\\

Although several authors have touched on the issue of time-varying treatment effects, we see several major gaps in the literature. First, no one has characterized the behavior of the standard treatment effect estimator when the assumption of an immediate treatment effect is violated. Second, there is ambiguity and lack of precise terminology around estimands of interest. Third, different models that account for time-varying treatment effects have not been studied in a unified manner.\\

This paper is organized as follows. In section 2, we analytically examine the behavior of the class of models that assume an immediate treatment effect and show that they may exhibit counterintuitive behavior in the presence of time-varying treatment effects. In section 3, we define several potential estimands of interest and discuss the assumptions that a researcher might be willing to make about the shape of the effect curve, including smoothness, monotonicity, and time until the maximum effect is reached. We then introduce a new set of models that account for time-varying treatment effects. One model accounts for the effect curve using fixed effects corresponding to exposure time by treatment interaction terms and other models do so using spline terms. We discuss how these models can be used for estimation and hypothesis testing in the context of both confirmatory and exploratory analyses. In section 4, we perform a simulation study to test the behavior of the models across a selection of data-generating mechanisms. In section 5, we illustrate the use of these models in real datasets. In section \ref{sec_discussion}, we discuss these results and their implications for the analysis of stepped wedge trials.\\

\section{Behavior of the IT model under a time-varying treatment effect}\label{sec_behavior}

In this section, we analytically examine the behavior of the IT model treatment effect estimator when the true treatment effect varies as a function of exposure time. Throughout this paper, we assume that the data come from a ``standard'' stepped-wedge trial, meaning that there are $J$ equally-spaced time steps, $Q=J-1$ sequences, and an equal number of clusters per sequence, with cross-sectional measurements taken at each time point. Assume the outcome data are continuous, and let $Y_{ijk}$ denote the observed outcome for individual $k \in \{1,...,K\}$ within cluster $i \in \{1,...,I\}$ at time $j \in \{1,...,J\}$. Also let sequences be labeled such that in sequence $q$, the treatment is introduced at time point $j=q+1$. For simplicity, we focus our analysis on the Hussey and Hughes model\cite{hussey2007design}, a special case of the IT model that is commonly used to analyze data from stepped wedge trials, although the behavior we observe holds more generally. This model is specified in (\ref{eq_husseyhughes}), where the $\beta_j$ terms represent the underlying time trend (with $\beta_1=0$ for identifiability), $\alpha_1,...,\alpha_I \overset{iid}{\sim} N(0,\tau^2)$ is a set of random effects accounting for the dependence of observations within a cluster, and $\epsilon_{111},...,\epsilon_{IJK} \overset{iid}{\sim} N(0,\sigma^2)$ are residual error terms.

\begin{equation}\label{eq_husseyhughes}
\begin{aligned}
	Y_{ijk} &= \mu + \beta_j + \delta X_{ij} + \alpha_i + \epsilon_{ijk} \\
	X_{ij} &= \begin{cases}
		1, & \text{cluster } i \text{ is in the treatment state at time } j \\
		0, & \text{otherwise}
	\end{cases}
\end{aligned}
\end{equation}\\

Next, let $s_{ij}$ represent the exposure time of cluster $i$ at time $j$ and let $\delta(s)$ represent the treatment effect at exposure time $s$, which we refer to as the ``point treatment effect'' (PTE) at $s$. Model (\ref{eq_datagen}) incorporates the time-varying treatment effect $\delta(s_{ij})$, where $X_{ij}$ is defined as above:

\begin{equation}\label{eq_datagen}
	Y_{ijk} = \mu + \beta_j + \delta(s_{ij}) X_{ij} + \alpha_i + \epsilon_{ijk}
\end{equation}\\

When the treatment effect varies with time, as in model (\ref{eq_datagen}), we may be interested in summary parameters such as the \textit{time-averaged treatment effect} (the average value of the effect curve over some period of time) or the \textit{long-term effect}; we will later define these terms more formally.\\

Suppose that data are generated according to (\ref{eq_datagen}) but analyzed with the model specified in (\ref{eq_husseyhughes}). We are interested in how the estimator $\hat{\delta}$ behaves in this scenario. Since $s_{ij}$ takes on values between $1$ and $J-1$, it would not be unreasonable to expect the IT estimator to converge to a value close to the time-averaged treatment effect over the study period, lying between the smallest and largest of the time-specific treatment effects $\delta(1),...,\delta(J-1)$. However, this turns out to not necessarily be the case. This conclusion follows from Theroem \ref{thm_it}, in which we provide closed-form expressions for the treatment effect estimator resulting from model (\ref{eq_husseyhughes}) and its expectation.\\

\begin{theorem}\label{thm_it}
\textit{Suppose we have a standard stepped wedge design with data generated according to (\ref{eq_datagen}). If we define $\phi \equiv \tau^2/(\tau^2+\sigma^2/n)$ and denote the mean outcome in sequence $q \in \{1,...,Q\}$ at time point $j$ by $\bar{Y}_{qj}$, the treatment effect estimator $\hat{\delta}$ obtained by fitting model (\ref{eq_husseyhughes}) via weighted least squares can be expressed as:}

\begin{equation}
	\hat{\delta} = \frac{12(1+\phi Q)}{Q(Q+1)(\phi Q^2+2Q-\phi Q-2)}
	\sum_{q=1}^Q \sum_{j=1}^J \left[ Q I\{j>q\} -j+1 +\frac{\phi Q(2q-Q-1)}{2(1+\phi Q)} \right] \bar{Y}_{qj} 
\end{equation}\\

\textit{Furthermore, the expectation of the treatment effect estimator $\hat{\delta}$ can be expressed as:}

\begin{equation}\label{eq_it_expectation}
	E[\hat{\delta}] = \sum_{s=1}^{J-1} w(Q,\phi,s) \delta(s),
\end{equation}\\

\textit{where}

\begin{equation}\label{eq_weightfunction}
	w(Q,\phi,s) \equiv \frac{
	  6(s-Q-1)\left( (1+2\phi Q)s - (1+\phi+\phi Q)Q \right)
	}{
	  Q(Q+1)(\phi Q^2 + 2Q - \phi Q-2)
	}
\end{equation}\\

\end{theorem}

\begin{proof}
See Appendix \ref{app_proof}.
\end{proof}

The expectation $E[\hat{\delta}]$ is thus a weighted sum of the individual point treatment effects $\delta(1),...,\delta(J-1)$, and does not depend on the underlying study time trend. It is easily shown that $\sum_{s=1}^{J-1} w(Q,\phi,s) = 1$ for any choice of $(Q,\phi)$, which is necessarily the case since otherwise $\hat{\delta}$ would be biased when the IT model is correct. However, surprisingly and unfortunately, \textit{some of the weights can be negative}, depending on the values of $\phi$ and $Q$. Additionally, a simple calculation shows that $w(J-1)$, the weight corresponding to $\delta(J-1)$, will \textit{always} be negative (assuming $\phi>0$). Equation (\ref{eq_weightfunction}) is somewhat difficult to interpret in itself; the impact of this finding is perhaps best illustrated in Figure \ref{behavior}, which shows the effect curve estimated from model (\ref{eq_husseyhughes}) -- which is necessarily constant after the first exposure time point -- plotted against the true curve for four possible effect curves.\\

\begin{figure}[H]
\centerline{\includegraphics[width=6in]{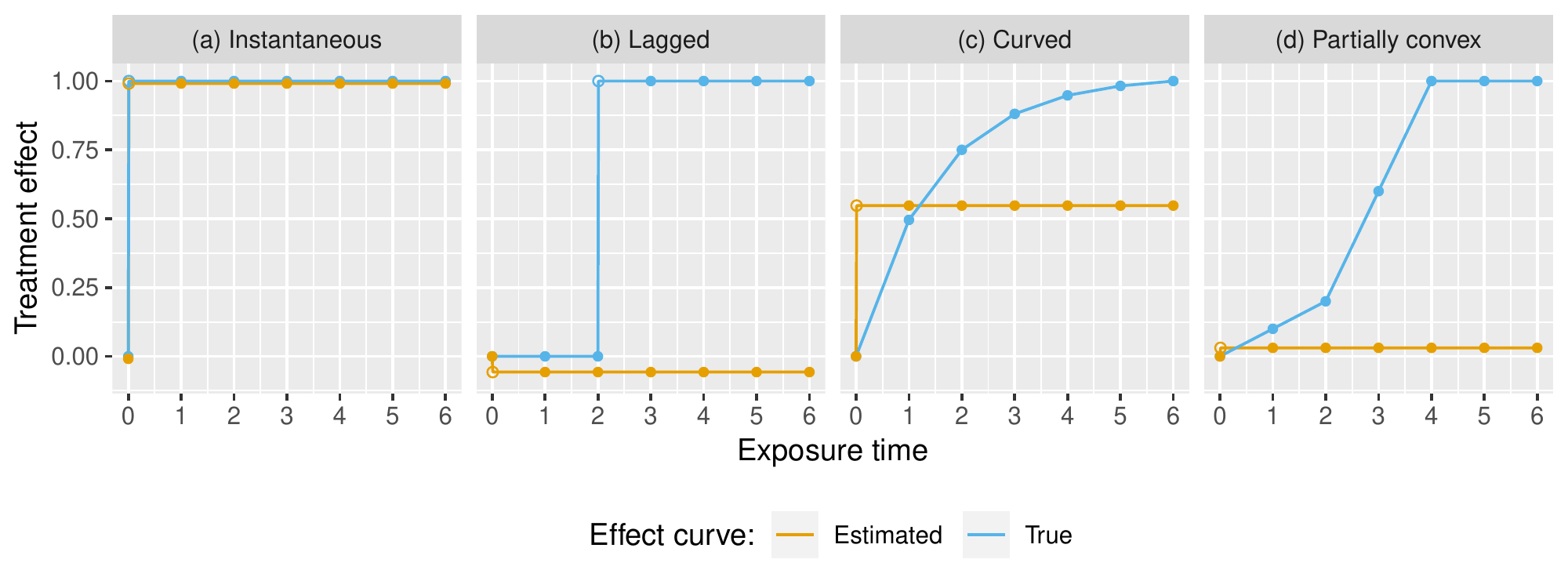}}
\caption{Four possible true effect curves plotted against the expected effect curves estimated from an IT model, for a design with $Q=6$ sequences and $\phi=0.5$\label{behavior}}
\end{figure}

Looking at panel (a) of Figure (\ref{behavior}), we see that when the assumption of an immediate treatment effect is correct, we correctly estimate the effect curve; this is expected, as weighted least squares estimators in linear mixed models are unbiased in general. However, the results from panels (b), (c), and (d) show that when the immediate treatment effect assumption is violated, the estimated effect curve can be astonishingly misleading. In panels (b) and (d), the estimated effect curve lies entirely below the true effect curve. Furthermore, we observe that even if each time-specific treatment effect is positive (or zero), the estimated treatment effect can actually converge to a negative value! This, in turn, implies that estimators of both the time-averaged treatment effect (over the study period) and the long-term effect can converge to a value of the opposite sign as the true value.\\

The value $\phi=0.5$ used to compute the estimates in Figure (\ref{behavior}) is not unreasonable for a stepped wedge design. As noted by Matthews and Forbes\cite{matthews2017stepped}, large numbers of individuals per cluster can lead to large values of $\phi$, even if the intraclass correlation coefficient (ICC) $\rho \equiv \tau^2/(\tau^2+\sigma^2)$ is small. For example, for an ICC of 0.05, $n=10$ leads to $\phi=0.34$ and $n=50$ leads to $\phi=0.72$. In Appendix \ref{app_corr}, we provide a figure showing how the weights vary as a function of $\phi$ for select values of $Q$, as implied by (\ref{eq_weightfunction}).\\

We also note that this counterintuitive behavior is not restricted to model (\ref{eq_husseyhughes}). In Appendix \ref{app_corr}, we conduct a numerical analysis to compute the weights corresponding to a model that includes the addition of a random time effect (i.e. a random intercept corresponding to a cluster-by-time interaction term, as in the Hooper/Girling model\cite{kasza2019impact}), such that the correlation matrix has a nested exchangeable structure.\cite{li2020mixed} We also perform the same analysis for a model involving a random treatment effect. For both alternative correlation structures, the resulting weights are qualitatively similar to those resulting from model (\ref{eq_husseyhughes}), in the sense that the largest weight is on $\delta(1)$ and some weights are negative. This analysis shows that the counterintuitive behavior of the IT model treatment effect estimator occurs across a wide range of models that are commonly used for analyzing stepped wedge trials.\\

Furthermore, this issue will also apply to models fit using generalized estimating equations (GEE). The results of Theorem \ref{thm_it} will apply to a GEE model with the fixed effects structure of (\ref{eq_husseyhughes}) and a working exchangeable correlation structure. If instead a working independence correlation structure is used, then instead of result (\ref{eq_it_expectation}), we will have $E[\hat{\delta}] = \sum_{s=1}^{J-1} w(Q,0,s) \delta(s)$ (i.e. the same result but with $\phi=0$). When $\phi=0$, there are no negative weights, but the weights will also not imply an estimand of direct scientific interest. For example, in a design with $Q=4$ sequences and $J=5$ time points, this will lead to $E[\hat{\delta}] = 0.6\delta(1) + 0.3\delta(2) + 0.1\delta(3) + 0.0\delta(4)$.\\

\section{Methods}

As demonstrated in section (\ref{sec_behavior}), a model that assumes an immediate treatment effect can lead to severely misleading inference if the treatment effect varies with exposure time. This necessitates the development of models that explicitly account for time-varying treatment effects. But before presenting these models, we define our estimands of interest more formally.\\

\subsection{Estimands of interest}

In this section, we formally introduce a model for the data-generating mechanism for the purpose of defining estimands of interest. Consider a standard stepped wedge design, as defined in section \ref{sec_behavior}. Let $Y_{ijk}(q)$ denote the potential outcome for individual $k \in \{1,...,K\}$ within cluster $i \in \{1,...,I\}$ at time $j \in \{1,...,J\}$, had this cluster been assigned to sequence $q \in \{1,...,J-1\}$, where again sequences are labeled such that in sequence $q$, the treatment is introduced at time point $j=q+1$. Also let $\mu_{ijk}(q)$ denote the expectation of $Y_{ijk}(q)$. We assume the following structural model for $\mu_{ijk}(q)$, an adaptation of the generic structural mixed model of Li et al.\cite{li2020mixed}:

\begin{equation}\label{eq_li}
	g \left( \mu_{ijk}(q) \right) = \overbrace{\Gamma_j}^\text{time trend} + \overbrace{\Delta_{ijq}}^\text{treatment effect} + \overbrace{C_{ijkq}}^\text{heterogeneity}
\end{equation}\\

Above, $g(\cdot)$ is a link function, and the outcome is assumed to fall within a parametric family with mean $\mu_{ijk}(q)$. $\Gamma_j$ represents the underlying study time trend and $\Delta_{ijq}$ represents the treatment effect. The $C_{ijkq}$ term accounts for the correlation structure of the data, and will generally take the form of one or more mean-zero random effect terms to capture cluster-level or temporal deviations from the fixed effects.\\

Let the random variable $Y_{ijk}$ denote the observed outcome. Denote the observed design matrix after randomization as $X \equiv \left\{ X_{ij} : i \in \{1,...,I\}, j \in \{1,...,J\} \right\}$, where $X_{ij}=1$ if cluster $i$ is is assigned to be in the treatment state at time $j$. If we assume that clusters are assigned to sequences randomly (exchangeability) and that $Y_{ijk}(q)=Y_{ijk}$ if cluster $i$ is randomly assigned to sequence $q$ (consistency), then the potential outcome expectation $\mu_{ijk}(q)$ is identified by $\mu_{ijk} \equiv E[Y_{ijk}|X_{ij}]$, and structural model (\ref{eq_li}) implies:

\begin{equation}\label{eq_li2}
	g \left( \mu_{ijk} \right) = \Gamma_j + \Delta_{ij}X_{ij} + C_{ijk}
\end{equation}\\

In the setting of randomized trials, these assumptions will typically hold. In equation (\ref{eq_li2}), the $q$ subscripts have been removed since they are determined by $i$ and $j$ given $X_{ij}$. Other works have focused on choices for the time trend $\Gamma_j$\cite{nickless2018mixed,hemming2017analysis,li2020mixed} and for the heterogeneity component $C_{ijk}$.\cite{li2020mixed,girling2016statistical,hooper2016sample,kasza2019impact,grantham2019accounting} In this paper, we restrict attention to modeling choices related to the treatment effect $\Delta_{ij}$. Furthermore, we assume that the treatment effect may vary as a function of exposure time, but \textit{not} as a function of study time. If $s_{ij}$ again represents the exposure time of cluster $i$ at time $j$, and $\delta$ is a fixed function, then this assumption allows us to simplify model (\ref{eq_li2}) as follows:

\begin{equation}\label{eq_eti}
	g \left( \mu_{ijk} \right) = \Gamma_j + \delta(s_{ij}) X_{ij} + C_{ijk}
\end{equation}\\

This model makes no assumptions about the form of the time trend or the heterogeneity, and is thus general enough to apply in a number of settings, such as when there are random time effects or random treatment effects.\\

In this context, referring to a single parameter as the ``treatment effect'' is ambiguous, since the treatment effect varies as a function of exposure time. Therefore, we must state more precisely what we are interested in estimating. In section \ref{sec_intro}, we heuristically introduced several estimands that we will now define more formally. First, we define the \textit{effect curve} as the function $s \mapsto \delta(s)$, where we think of time as continuous even though data collection for stepped wedge trials typically occurs at discrete time points. This function is an appealing estimand since it contains a wealth of information that can be used to guide program design and evaluation. However, in many applications it will be desirable to focus on a single parameter that summarizes this curve. The first summary estimand we consider is the time-averaged treatment effect (TATE) from exposure time $s_1$ to exposure time $s_2$, which is the average value of the effect curve over the interval $[s_1,s_2]$. We denote this estimand by $\Psi_{[s_1,s_2]}$ and formally define it as follows:

\begin{equation}\label{eq_tate}
	\Psi_{[s_1,s_2]} \equiv \frac{1}{s_2-s_1} \int_{s_1}^{s_2} \delta(s)ds
\end{equation}\\

For a given study that involves data collection at $J$ discrete time points, a common choice of the interval $[s_1,s_2]$ will be $[0,J-1]$, the maximum treatment period. A researcher may also choose to increase the lower endpoint if he or she is interested in the TATE after a washout period immediately following implementation of the treatment. Care should be taken when comparing TATE estimates from two different studies, even if they both examine the same intervention and use the same outcome measure, since the TATE is defined relative to a time interval, and these time intervals may differ between studies.\\

The second summary estimand we consider is the point treatment effect (PTE) at exposure time $s_0$, which is simply defined as the value $\Psi_{s_0} \equiv \delta(s_0)$. A common choice for $s_0$ might be $J-1$, the maximum treatment period. If we additionally assume that the effect curve ``flattens out'' by the end of the study (that is, we have that $\delta(s) = \delta(J-1)$ whenever $s \ge J-1$), $\Psi_{J-1}$ can be interpreted as the long-term treatment effect (LTE).\\

When the true effect curve has no delay or change over time (i.e. $\delta(s)$ is constant), we have that $\Psi_{[s_1,s_2]}=\Psi_{s_0}$ for all $s_0,s_1,s_2$. However, in the presence of a time-varying effect, these estimands do not in general coincide, and so in the context of a confirmatory trial, the researcher must specify one in advance as the primary endpoint. Others can be examined within secondary or exploratory analyses, if applicable. We expect that researchers will most often be interested in the TATE or the LTE as the target for confirmatory trials, whereas estimation of the effect curve $s \mapsto \delta(s)$ will usually be done as an exploratory analysis. Since different researchers will be interested in different estimands, we will consider estimation of all three in this paper.\\

\subsection{Analysis models}\label{sec_analysismodels}

Model (\ref{eq_eti}) was written generically so that the estimands are well-defined in a variety of settings. The analysis models we present in this section can be seen as special cases of (\ref{eq_eti}) that make assumptions about the treatment effect term $\delta(s)X_{ij}$. Implementation of these models in specific settings will involve making context-appropriate choices for the $\Gamma_j$ and $C_{ijk}$ terms. For example, we may choose to take $\Gamma_j = \mu + \beta_j$ and $C_{ijk} = \alpha_i$ with $\alpha_1,...,\alpha_I \overset{iid}{\sim} N(0,\tau^2)$ to yield a class of models containing the Hussey and Hughes model.\cite{hussey2007design} All models in this section can be used for the purposes of estimation and hypothesis testing related to both the TATE and the PTE/LTE.\\

\textbf{Immediate treatment (IT) model}

First, we consider the \textit{immediate treatment} (IT) model, a generalization of the model we studied in section \ref{sec_behavior}. This model assumes that full effect is reached within a single time step and does not change thereafter, as in Figure \ref{effect_curves}a; we provide a specification of this model in equation (\ref{eq_it}), where $\delta$ is a scalar parameter. Most mixed models that have previously been used to analyze data from stepped wedge trials, such as the Hussey and Hughes model\cite{hussey2007design} and the Hooper/Girling model\cite{kasza2019impact}, are special cases of the IT model.\\

\begin{equation}\label{eq_it}
\begin{aligned}
	g \left( \mu_{ijk} \right) &= \Gamma_j + \delta X_{ij} + C_{ijk} \\
	X_{ij} &= \begin{cases}
		1, & \text{cluster } i \text{ is in the treatment state at time } j \\
		0, & \text{otherwise}
	\end{cases}
\end{aligned}
\end{equation}\\

As noted, when the IT model is correct (i.e. there is no variation in the treatment effect over time), the TATE (for any $[s_1,s_2]$) and PTE (for any $s_0$) coincide, and so we can estimate both by fitting this model and using $\hat{\delta}$ as our estimator. When the IT model is correct, $\hat{\delta}$ is a more efficient estimator of the TATE or PTE than the other estimators we will consider. However, when the IT model is incorrect, $\hat{\delta}$ can give highly misleading results, as described in section \ref{sec_behavior}.\\

\textbf{Exposure time indicator (ETI) model}

Next, we consider a model that makes no assumptions about the shape of the effect curve, which we refer to as the \textit{exposure time indicator} (ETI) model. This is generalization of models considered in Granston et al.\cite{granston2014addressing} and Nickless et al.\cite{nickless2018mixed} and is specified by equation (\ref{eq_eti2}), where, with a slight abuse of notation, we use the subscript $s$ as shorthand for $s_{ij}$, the exposure time of cluster $i$ at time $j$.

\begin{equation}\label{eq_eti2}
\begin{aligned}
	g \left( \mu_{ijk} \right) &= \Gamma_j + \delta_s X_{ij} + C_{ijk}
\end{aligned}
\end{equation}\\

Model (\ref{eq_eti2}) is saturated with respect to exposure time; as such, it is the most flexible model we consider. After fitting this model and obtaining parameter estimates $\hat{\delta}_1,...,\hat{\delta}_{J-1}$, we can base estimation and inference on the following estimators, where we assume that $s_0$, $s_1$, and $s_2$ are all integers in the set $\{0,...,J-1\}$ with $s_1<s_2$.

\begin{equation}
\begin{aligned}\label{eq_est1}
	\hat{\Psi}_{[s_1,s_2]} &\equiv \frac{1}{s_2-s_1} \sum_{r=1}^{s_2-s_1} \hat{\delta}_{s_1+r} \\
	\hat{\Psi}_{s_0} &\equiv \hat{\delta}_{s_0}
\end{aligned}
\end{equation}\\

The estimator $\hat{\Psi}_{[s_1,s_2]}$ can be viewed as a right-hand Riemann sum that approximates the integral $\Psi_{[s_1,s_2]}$. Alternatively, a researcher can use an estimator based on a trapezoidal Riemann sum; this is discussed in See Appendix \ref{app_trapezoid}. The estimator $\hat{\Psi}_{J-1}$ can be used to estimate the LTE if it is assumed to exist, although as we will see, the variance of this estimator will typically be quite high. Additionally, the entire effect curve can be estimated as a linear spline or step function based on $\hat{\delta}_1,...,\hat{\delta}_{J-1}$.\\

A key assumption of this model is that the shape and height of the effect curve do not vary between clusters; this is a fairly strong assumption that we will later relax when we discuss random treatment effects in section \ref{sec_rte}. We also note that the IT model is a submodel of the ETI model, and so (likely in the context of an exploratory analysis), a researcher can conduct a likelihood ratio test to determine whether the IT model is an appropriate simplification of the ETI model for a given dataset.\\

Since $\hat{\Psi}_{[s_1,s_2]}$ is a linear combination of the parameter estimates $\hat{\delta}_1,...,\hat{\delta}_{J-1}$, variance estimation is straightforward. Assuming the $(J-1) \times 1$ matrix $\hat{\delta} \equiv (\hat{\delta}_1,...,\hat{\delta}_{J-1})^T$ is approximately multivariate Normal with estimated covariance matrix $\hat{V}$, we find the $1 \times (J-1)$ matrix $M$ such that $\hat{\Psi}_{[s_1,s_2]} = M \hat{\delta}$. Our estimated variance is then $\hat{\text{Var}}(\hat{\Psi}_{[s_1,s_2]}) \equiv M\hat{V}M'$, where $M'$ represents the transpose of $M$. For example, if we are estimating $\hat{\Psi}_{[0,J-1]}$, we have that $M = \left(\frac{1}{J-1},...,\frac{1}{J-1} \right)$ and $M\hat{V}M'=\frac{1}{(J-1)^2}\sum_{i=1}^{J-1}\sum_{j=1}^{J-1} \hat{V}_{ij}$.\\

\textbf{Restricted exposure time indicator (RETI) model}

Sometimes, a researcher may have contextual information about the shape of the effect curve that can be leveraged to achieve increased precision. In particular, it may be known (or assumed) that the effect curve ``flattens out'' after some time $s^*$. If we make this assumption, then we can use the following adaptation of the ETI model (\ref{eq_eti2}), which we refer to as the \textit{restricted exposure time indicator} (RETI) model. The IT model can be seen as the special case of the RETI model in which $s^*=0$. As we will see, this model will be particularly useful when the estimand of interest is the LTE.

\begin{equation}\label{eq_reti}
\begin{aligned}
	g \left( \mu_{ijk} \right) &= \Gamma_j + \delta_{\min\{s,s^*\}} X_{ij} + C_{ijk}
\end{aligned}
\end{equation}\\

Intuitively, this model pools the data for all time points greater than or equal to $s^*$ to increase precision. Smaller values of $s^*$ will lead to increased precision, but at the cost of stronger modeling assumptions. Once we have fit this model and obtained parameter estimates $\hat{\delta}_1,...,\hat{\delta}_{s^*}$, we can base estimation and inference on the following estimators:

\begin{equation}
\begin{aligned}\label{eq_est2}
	\hat{\Psi}_{[s_1,s_2]} &\equiv \begin{cases}
		\hat{\delta}_{s^*}, & s^* \le s_1 \\
		\frac{1}{s_2-s_1} \left( \sum_{r=1}^{s^*-s_1} \hat{\delta}_{s_1+r} +
		(s_2-s^*)\hat{\delta}_{s^*} \right), & s^* \in (s_1,s_2) \\
		\frac{1}{s_2-s_1} \sum_{r=1}^{s_2-s_1} \hat{\delta}_{s_1+r}, & s^* \ge s_2
	\end{cases}\\
	\hat{\Psi}_{s_0} &\equiv \hat{\delta}_{\min\{s_0,s^*\}}
\end{aligned}
\end{equation}\\

As in the ETI model, since $\hat{\Psi}_{[s_1,s_2]}$ is a linear combination of the parameter estimates, variance estimation is straightforward.\\

\textbf{Natural cubic spline (NCS) model}

If we are willing to assume some degree of smoothness of the effect curve, we may be able to construct more precise estimates. This suggests replacing the $\delta(s)X_{ij}$ term in (\ref{eq_eti}) with some sort of smoothing term. In this paper, we consider the use of natural cubic splines, although other approaches are possible. For an overview of spline-based methods, see Friedman et al.\cite{friedman2001elements} Briefly, given a set of real numbers $k_1<...<k_d$ called ``knots'', a cubic spline is a function that is equivalent to a cubic polynomial over any interval $[k_r,k_{r+1}]$ for $r \in \{1,...,d-1\}$. A natural cubic spline is further constrained such that it is continuous, twice continuously differentiable, and linear to the left of the first knot and to the right of the last knot. Through the construction of a so-called spline basis -- a set of functions $b_1,...,b_d$ that are applied to the variable of interest -- natural cubic splines can be embedded within the linear mixed model framework. Typically, a natural cubic spline with $d$ knots can be represented with $d$ basis functions. In our context, we enforce the additional constraint that the spline must pass through the origin. This leads us to consider the following model, in which the terms $b_1(s), ..., b_d(s)$ represent a natural cubic spline basis with $d$ degrees of freedom.

\begin{equation}\label{eq_ncs}
	g \left( \mu_{ijk} \right) = \Gamma_j + \left[ \omega_1 b_1(s) + ... + \omega_d b_d(s) \right] X_{ij} + C_{ijk}
\end{equation}\\

Note that the construction of the spline basis requires the user to specify the number and placement of the knots. A model that uses a basis with $J-1$ degrees of freedom will yield identical estimation and inference to the ETI model, and so if this model is used, it should be the case that $d<J-1$. Given that fitting this model yields an estimate of the entire function $s \mapsto \delta(s)$, we could in theory estimate the TATE via integration. However, for simplicity and to achieve consistency with the other models we consider, we again use a right-hand Riemann sum approximation. This allows us to conduct estimation and inference based on the estimates $\hat{\omega}_1,...,\hat{\omega}_d$ of $\omega_1,...,\omega_d$ and the associated covariance matrix estimate. This yields the following estimators, which are analogs of (\ref{eq_est1}):

\begin{equation}
\begin{aligned}\label{eq_est3}
	\hat{\Psi}_{[s_1,s_2]} &\equiv \frac{1}{s_2-s_1} \sum_{r=1}^{s_2-s_1} \sum_{k=1}^d \hat{\omega}_k b_k(s_1+r) \\
	\hat{\Psi}_{s_0} &\equiv \sum_{k=1}^d \hat{\omega}_k b_k(s_0)
\end{aligned}
\end{equation}\\

The spline basis and the estimators above are easily calculable with existing software, and variance estimation proceeds as before, since these estimators are also linear combinations of the parameter estimates.\\

\textbf{Monotone effect curve (MEC) model}

It will often be reasonable to assume that the effect curve is monotone; for example, more cardiovascular exercise sustained over a longer period of time will generally lead to more weight loss. When this assumption is  plausible, it is natural to wonder whether we can leverage this knowledge to obtain more precise estimates. Constructing models that enforce monotonicity can be done in a number of ways; we introduce one possible model here.

\begin{equation}\label{eq_mon}
\begin{aligned}
	g \left( \mu_{ijk} \right) &= \Gamma_j + \delta \left( \sum_{t=0}^{s} \alpha_t \right) X_{ij} + C_{ijk} \\
\end{aligned}
\end{equation}\\

Above, we let $\alpha_0 \equiv 0$ and constrain $(\alpha_1,...,\alpha_{J-1})$ as a simplex such that $\alpha_t \ge 0$ for $t \in \{1,...,J-1\}$ and $\sum_{t=1}^{J-1} \alpha_t=1$. This model parameterizes the effect curve as a monotonic step function, but allows for the step function to be either nondecreasing or nonincreasing. This is a constrained nonlinear model, and so more advanced techniques are needed to fit it and estimate the parameters $(\delta,\alpha_1,...,\alpha_{J-1})$. We choose to fit it as a hierarchical Bayesian model using the following prior specification, where $(c_1,...,c_{J-1})$ are fixed constants and $(\delta,\omega,\alpha_1,...,\alpha_{J-1})$ are parameters:

\begin{equation}\label{eq_mon2}
\begin{aligned}
	\delta &\sim \text{Normal}(0,10^4) \\
	\omega &\sim \text{Uniform}(0.01,100) \\
	(\alpha_1,...,\alpha_{J-1}) &\sim \text{Dirichlet}(c_1\omega,...,c_{J-1}\omega) \\
\end{aligned}
\end{equation}\\

Choosing some values of the constants $(c_1,...,c_{J-1})$ to be larger than the others can be seen as encoding the prior belief that the biggest ``jump'' in the effect curve will occur in a certain region of the effect curve, or equivalently that the curve will remain relatively flat over a certain region. For example, choosing $c_1=...=c_{(J-1)/2}=5$ and $c_{(J-1)/2+1}=...=c_{J-1}=1$ encodes the prior belief that the biggest jump in the effect curve occurs in the first half of the effect curve (assuming $(J-1)/2$ is an integer); this is the prior we use in this paper in the simulations and the real data analyses. Choosing $c_1=...=c_{J-1}=1$ leads to a symmetric Dirichlet prior, which can be seen as minimally informative in the sense that it doesn't presume in advance that the jump happens at any particular time point.\\

It is important to note that monotone-constrained function estimators tend to be biased at the endpoints\cite{shively2009bayesian}; in our case, the endpoint of the effect curve is precisely the LTE, if it is assumed to exist. The prior described above has the effect of counteracting this bias to an extent by stabilizing the tail of the effect curve. Thus, the results we obtain when using this model are influenced not just by the monotonicity constraint, but also by the informative prior. As we will see in both simulated and real data, different choices of prior may lead to very different estimates, and so we recommend that this model only be used in the context of an exploratory analysis.\\

In this model, the LTE estimator is simply the posterior mean $\hat{\delta}$. For the TATE, we again use an estimator based on a right-hand Riemann sum. This can be done by computing the posterior means $\hat{\alpha}_1,...,\hat{\alpha}_{J-1}$, calculating $\hat{\delta}_1,...,\hat{\delta}_{J-1}$ using the formula $\hat{\delta}_s \equiv \sum_{t=1}^s \hat{\alpha}_t$, and then calculating the estimators in (\ref{eq_est1}) or (\ref{eq_est1b}). Variance estimation can be done using standard methods for Bayesian inference.\\

\subsection{Incorporating random effects}\label{sec_rte}

In the analysis of stepped wedge trials using an immediate treatment model, it is common to include mean-zero random effect terms to capture cluster-level or temporal deviations from the fixed effects. For example, we may include random intercepts to model cluster-level deviations, and we may include random effects at the level of the cluster-by-time interaction (a type of ``random time effect'') to allow the underlying time trend to differ by cluster. These random effects can be incorporated into models that allow for a time-varying treatment effect in the same way they would be added to a model that assumes an immediate treatment effect. However, more thought is required in terms of incorporating ``random treatment effects'', which allow for the effect of the treatment to vary between clusters. When the treatment effect does not vary with time, this takes the form of a a random coefficient on the treatment indicator variable, indexed by cluster. When the effect of treatment varies by exposure time, there are multiple ways in which random effects could be used to allow the effect curve to differ across clusters. For simplicity, we focus on the case in which the only heterogeneity component other than the random treatment effect is a cluster-level random intercept. The specification given in (\ref{eq_eti2re}) allows for the ``height'' of the effect curve to vary between clusters and involves two additional parameters. Essentially, each of the parameters $(\delta_1,...,\delta_{J-1})$ represents the value of the effect curve at a particular exposure time averaged across clusters, and each of the parameters $(\eta_1,...,\eta_I)$ represents the amount by which the entire effect curve is vertically shifted for a given cluster. Other random effects structures that allow for both the ``height'' and the ``shape'' of the effect curve to vary between clusters are possible but may be difficult to fit due to their complexity.

\begin{equation}\label{eq_eti2re}
\begin{gathered}
	g \left( \mu_{ijk} \right) = \Gamma_j + (\delta_s+\eta_i) X_{ij} + \alpha_i \\
	\begin{pmatrix} \alpha_1 \\ \eta_1 \end{pmatrix} ,...,
	\begin{pmatrix} \alpha_I \\ \eta_I \end{pmatrix}
	\overset{iid}{\sim} N \left[
		\begin{pmatrix} 0 \\ 0 \end{pmatrix},
		\begin{pmatrix} \tau^2 & \rho \tau \nu \\ \rho \tau \nu & \nu^2 \end{pmatrix}
	\right]
\end{gathered}
\end{equation}\\

The RETI and NCS models can be adapted to include random treatment effects in an analogous manner.\\

\subsection{Implications for study design}\label{study_design}

In stepped wedge studies, data collection typically stops after all clusters have reached the treatment state, since additional data collection provides little additional information on the treatment effect when an immediate treatment model with a saturated study time effect is used for analysis. However, when a model is used that allows for time-varying treatment effects, additional data collection may lead to gains in precision, depending on the estimand of interest. We can gain intuition for why this is the case by considering estimation of the LTE. With a typical stepped wedge design, only one sequence is observed at exposure time $J-1$ (at study time $J$), and so if the ETI model is used, all the information about the change in the effect curve between times $J-1$ and $J-1$ must come from this sequence alone. However, if we collect additional data at study time step $J+1$, there are now two sequences that have been observed at exposure time $J-1$, and so we should be able to estimate the LTE more precisely, assuming we still think that the long-term effect is reached by exposure time $J-1$. Similar logic applies to estimation of the TATE, although we would not expect the gains in precision to be as large. For either estimand, the magnitude of this potential gain in precision can be roughly quantified via simulation.\\

\section{Simulation study}

We conducted a simulation study designed to assess the performance of the models described in section \ref{sec_analysismodels} in a variety of settings. Unless otherwise specified, data were generated according to the following model, a special case of (\ref{eq_eti}):

\begin{equation}
\begin{gathered}
	Y_{ijk} = \mu + \beta_j + \delta h(s) X_{ij} + \alpha_i + \epsilon_{ijk} \\
	(\alpha_1,...,\alpha_I) \overset{iid}{\sim} N(0,\tau^2) \\
	(\epsilon_{111},...,\epsilon_{IJK}) \overset{iid}{\sim} N(0,\sigma^2)
\end{gathered}
\end{equation}\\

The function $s \mapsto h(s)$ represents the effect curve, constrained to start at $(0,0)$ and achieve a maximum value of $1$ on the interval $[0,J-1]$. For all simulations, we generated data according to four different effect curves, defined as step function approximations to the effect curves a-d in Figure \ref{effect_curves}. This was done so that when we are comparing TATE estimates to the true values, we eliminate the component of the error due to the step function being an approximation to a smooth curve.\\

Simulations involved $I=24$ clusters, $J=7$ time points, and $K=20$ individuals per cluster. Parameters were fixed at $\mu=1$, $\delta=0.5$, $\sigma=2$, and $\tau=0.5$. This results in an ICC of $\rho=0.059$ and a value of $\phi$ (as defined in Theorem \ref{thm_it}) of $\phi=0.56$. The study design was balanced and complete, in the sense that there were an equal number of clusters in each sequence and one sequence crossed over at each time point. Data were generated with a linear time trend with $\beta_j = -0.5\left(\frac{j-1}{J-1}\right)$, although the analysis models treat time as categorical.\\

For the MEC model, we use a $\text{Dirichlet}(5\omega,5\omega,5\omega,\omega,\omega,\omega)$ prior. As discussed in section \ref{sec_analysismodels}, this can be seen as encoding the prior belief that the biggest ``jump'' in the effect curve will occur in the first three exposure time points, or equivalently that it will remain relatively flat after the last three exposure time points.\\

In all simulations, we considered estimation of $\Psi_{[0,J-1]}$, the TATE between exposure times $0$ and $J-1$, as well as the LTE $\Psi_{J-1}$, the value of the effect curve at time $J-1$. Performance was assessed by estimating bias, 95\% confidence interval coverage, and mean squared error (MSE). We also assessed estimation of the entire effect curve by calculating the average pointwise MSE, defined as $\frac{1}{J-1} \sum_{s=1}^{J-1} (\hat{\Psi}_s-\Psi_s)^2$, and estimated the power of Wald-type hypothesis tests. Each individual statistic was calculated using 1,000 simulation replicates. All simulations were conducted using the R programming language and structured using the \textit{SimEngine} simulation framework.\cite{kennysimengine2021}\\

\subsection{Performance of the IT, ETI, NCS, and MEC models}\label{sec_sim23}

First, we compare the performance of several of the models described above that are designed to account for a time-varying treatment effect, in addition to the IT model for comparison. Results are given in Figure \ref{sim2}. As expected, the IT model performs well when it is correct; simulation results confirm that the TATE/LTE estimate is unbiased and that nominal 95\% coverage is achieved. However when the IT model is not correct, estimates are severely biased and confidence interval coverage is unacceptable. The ETI model will always be correct in the sense that it makes no assumptions about the form of the effect curve. For this reason, it will also typically be the least efficient. We indeed observe that the ETI model leads to unbiased estimates of both the TATE and the LTE; the tiny amount of bias observed is due to Monte Carlo error. The NCS model performs similarly to the ETI model in terms of coverage and MSE, but appears to do slightly worse in terms of bias. The similar performance is not surprising in this case, as there are only six exposure time points in this simulation setup, and so the NCS model involves just two fewer degrees of freedom than the ETI model. The MEC model has the highest bias seen across effect curves but also consistently yields improved MSE for LTE estimation; this suggests that the model induces a bias-variance trade-off that may be beneficial depending on a researcher's goals. The MEC model performs similarly to the ETI model for TATE estimation. Note that ``coverage'' refers to credible interval coverage in the case of the MEC model.\\

\begin{figure}[H]
\centerline{\includegraphics[width=6in]{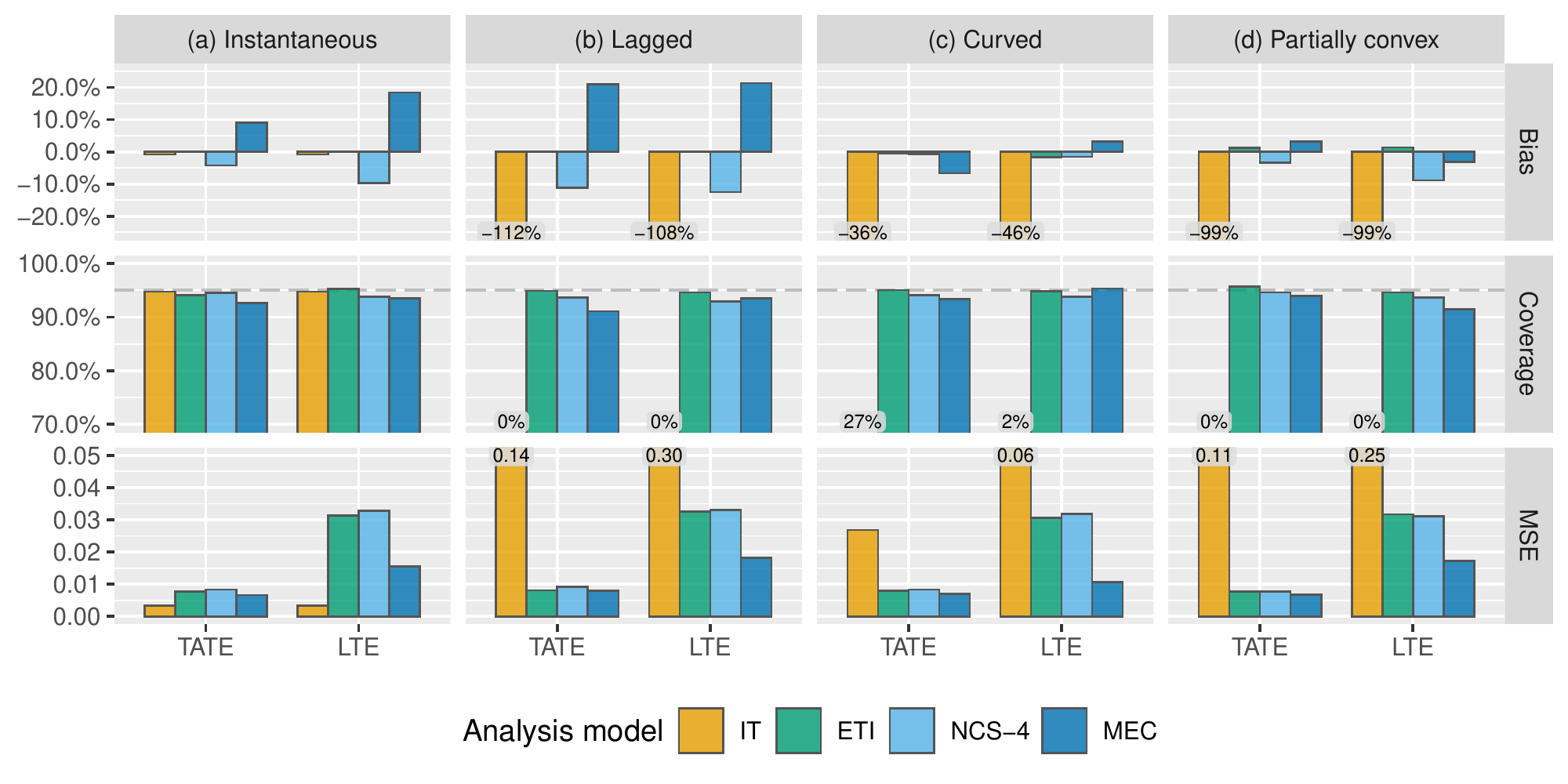}}
\caption{Simulation results: bias, coverage, and mean squared error (MSE) for the estimation of the TATE ($\Psi_{[0,J-1]}$) and LTE ($\Psi_{J-1}$) using the following four models: immediate treatment (IT), exposure time indicator (ETI), natural cubic spline with 4 degrees of freedom (NCS-4), monotone effect curve (MEC). Numeric values displayed over graph bars represent the height of the bars that are cut off because of the scale of the axes.}
\label{sim2}
\end{figure}

Figure \ref{sim3} shows the results of estimating the entire effect curve, assessed via average pointwise MSE, as described above. Overall, the ETI and NCS models perform similarly. The MEC model performs the best for three out of the four effect curves.\\

\begin{figure}[H]
\centerline{\includegraphics[width=6in]{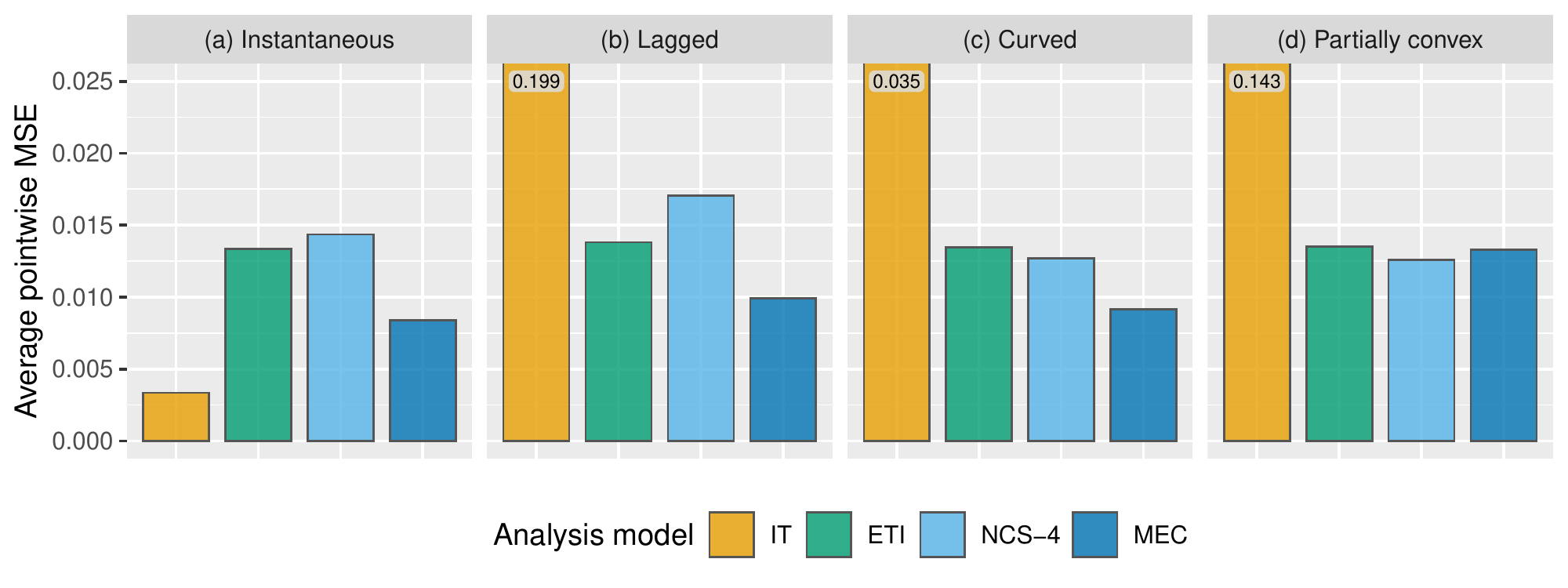}}
\caption{Simulation results: average pointwise mean squared error for the estimation of the entire effect curve using the following four models: immediate treatment (IT), exposure time indicator (ETI), natural cubic spline with 4 degrees of freedom (NCS-4), monotone effect curve (MEC). Numeric values displayed over graph bars represent the height of the bars that are cut off because of the scale of the axes.}
\label{sim3}
\end{figure}

Next, using the same set of models as above, we conduct a set of Wald-type hypothesis tests related to the TATE ($H_0 : \Psi_{[0,J-1]}=0$ vs. $H_1 : \Psi_{[0,J-1]} \ne 0$) and to the LTE ($H_0 : \Psi_{J-1}=0$ vs. $H_1 : \Psi_{J-1} \ne 0$). Results are given in Figure (\ref{sim4}). As expected, tests based on the IT model parameter estimates are the most powerful when the model is correct. When it is incorrect, the power of the IT model suffers dramatically for the lagged (b) and partially convex (d) effect curves. Importantly, this implies that if the IT model is used in the presence of a time-varying effect, there is a serious risk that hypothesis tests will suffer from high type II error rates even if a strong effect is present. The ETI and NCS models both perform reasonably well in all scenarios. All tests maintain proper type I error rates. However, because of the effects of model misspecificaion, there exist effect curves for which the IT model would lead to inflated type I error rates, although we do not demonstrate that here.\\

\begin{figure}[H]
\centerline{\includegraphics[width=6in]{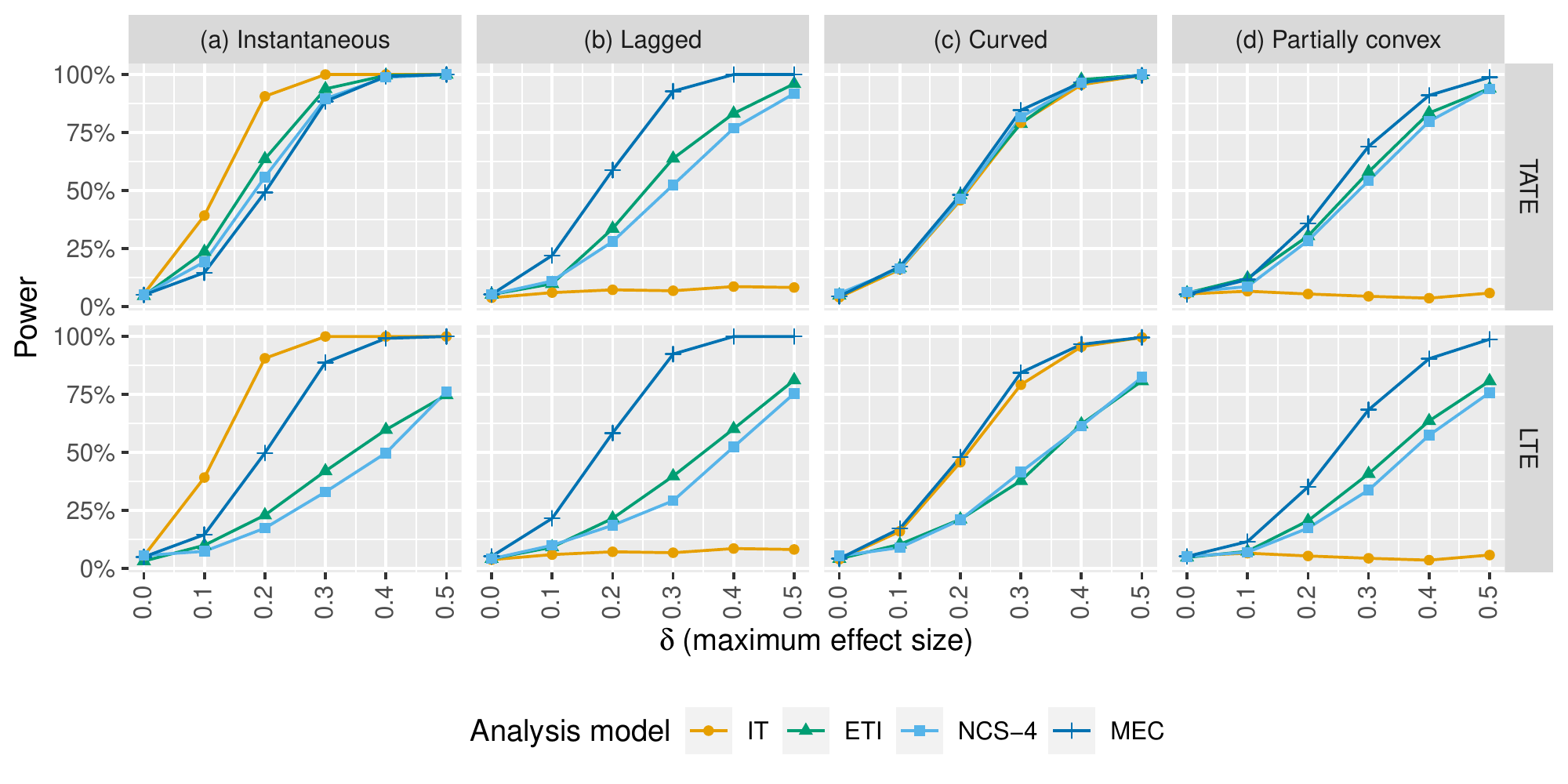}}
\caption{Simulation results: power of Wald-type hypothesis tests for testing null hypotheses related to the TATE ($\Psi_{[0,J-1]}=0$) and the LTE ($\Psi_{J-1}=0$) using the following four models: immediate treatment (IT), exposure time indicator (ETI), natural cubic spline with 4 degrees of freedom (NCS-4), monotone effect curve (MEC)}
\label{sim4}
\end{figure}

In Appendix \ref{app_additional_sims}, we compare the ETI model to an analogous model that additionally includes a random treatment effect term corresponding to the specification given in (\ref{eq_eti2re}). We show that the use of a random treatment effect in the analysis model is beneficial in the sense that performance does not suffer when no random treatment effect is present in the data, whereas if it is present in the data, performance improves considerably; this is consistent with previous findings.\cite{li2020mixed}\\

\subsection{Performance of RETI models}

Next, we test the RETI models, which involve the assumption that the effect curve ``flattens out'' after a certain number of time points, which we refer to as $s^*$. We test the models corresponding to $s^*=3$ and $s^*=4$. This model is correct for effect curves (a) and (b), which flatten out after 1 and 3 time points, respectively. Effect curve (c) becomes flatter at each time point, but does not reach maximum effect until exposure time $6$; this gives us the opportunity to assess the performance of the RETI models under slight violations of the flattening assumption. Effect curve (d) flattens out after 4 time points, and so for this curve, the model corresponding to $s^*=4$ is correct but the model corresponding to $s^*=3$ is not.\\

Results are given in Figure (\ref{sim5}). As expected, this model yields substantial gains in precision when the estimand of interest is the LTE. We also observe modest precision gains when estimating the TATE. Under violations of the flattening assumption, we still observe better performance in terms of MSE across all scenarios considered, although bias and coverage suffer correspondingly.\\

\begin{figure}[H]
\centerline{\includegraphics[width=6in]{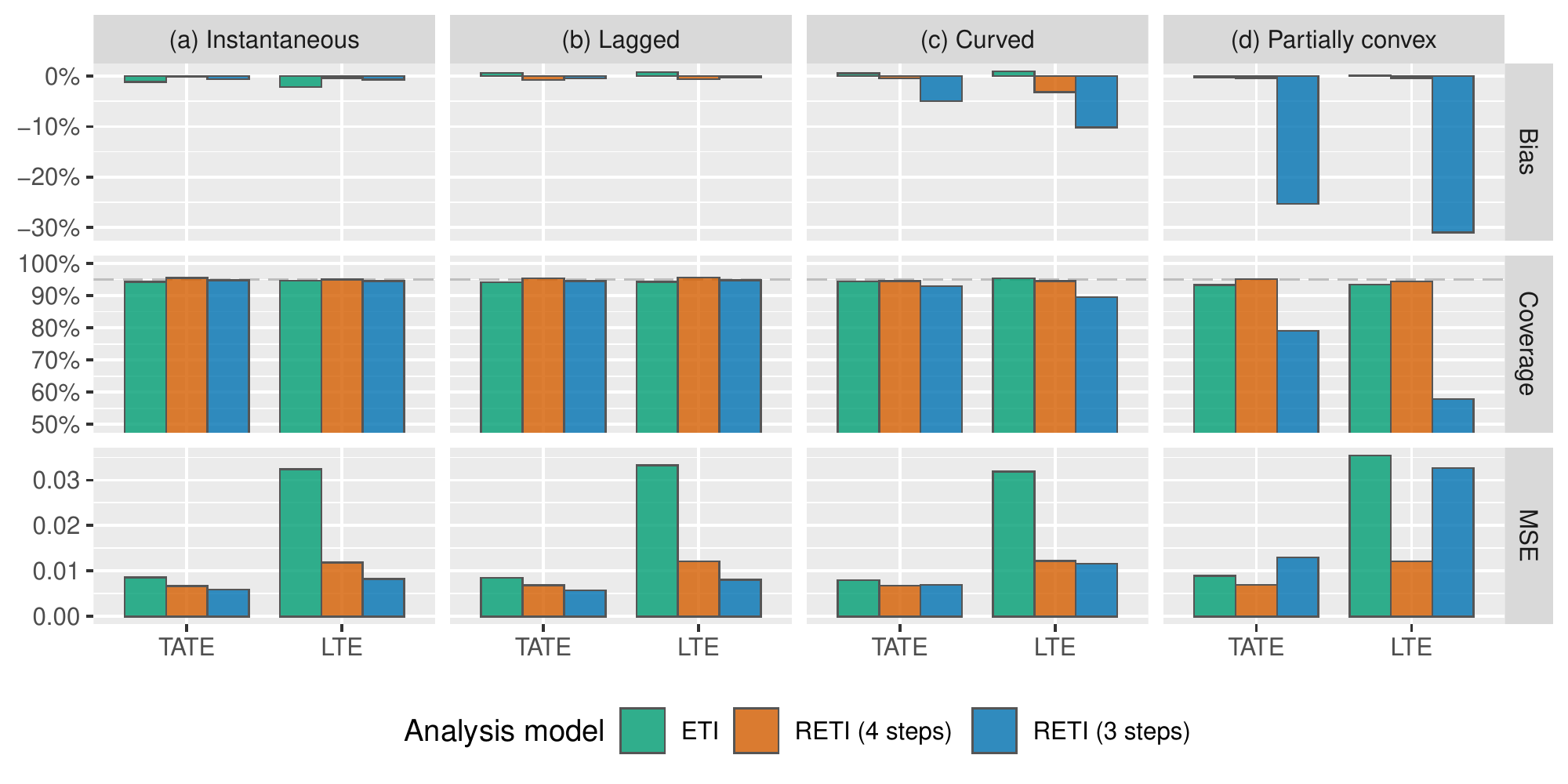}}
\caption{Simulation results: bias, coverage, and mean squared error (MSE) for the estimation of the TATE ($\Psi_{[0,J-1]}$) and LTE ($\Psi_{J-1}$) using the following three models: exposure time indicator (ETI), restricted exposure time indicator with flattening time $s^*=3$ (RETI-3), restricted exposure time indicator with flattening time $s^*=4$ (RETI-4)}
\label{sim5}
\end{figure}

\subsection{Effect of adding extra time points}

In the final set of simulations, we investigate the effect of adding additional data collection time points to the end of the study when estimating the TATE and LTE based on the ETI model. As discussed, when the IT model is used to analyze the data, it provides little benefit in terms of estimating the TATE or the LTE. However, when analyzing the data with one of the models that allows for time-varying treatment effects, precision gains may be possible, as discussed in section \ref{study_design}. Figure \ref{sim8} shows the results of simulations in which we added either zero, one, or two additional time points and analyze data using the ETI model. The effect curves used to generate the data are the same ones used in all previous simulations, but with $\Psi_{J-1}=\Psi_J=\Psi_{J+1}$; that is, the true effect curves remain flat after exposure time $J-1$. We continue to define the LTE as $\Psi_{J-1}$. We observe that adding extra time points enables considerable MSE gains in terms of estimating the LTE, but provides little benefit in terms of estimating the TATE. If the RETI model were used instead of the ETI model, the MSE gains would presumably be even greater.\\

\begin{figure}[H]
\centerline{\includegraphics[width=6in]{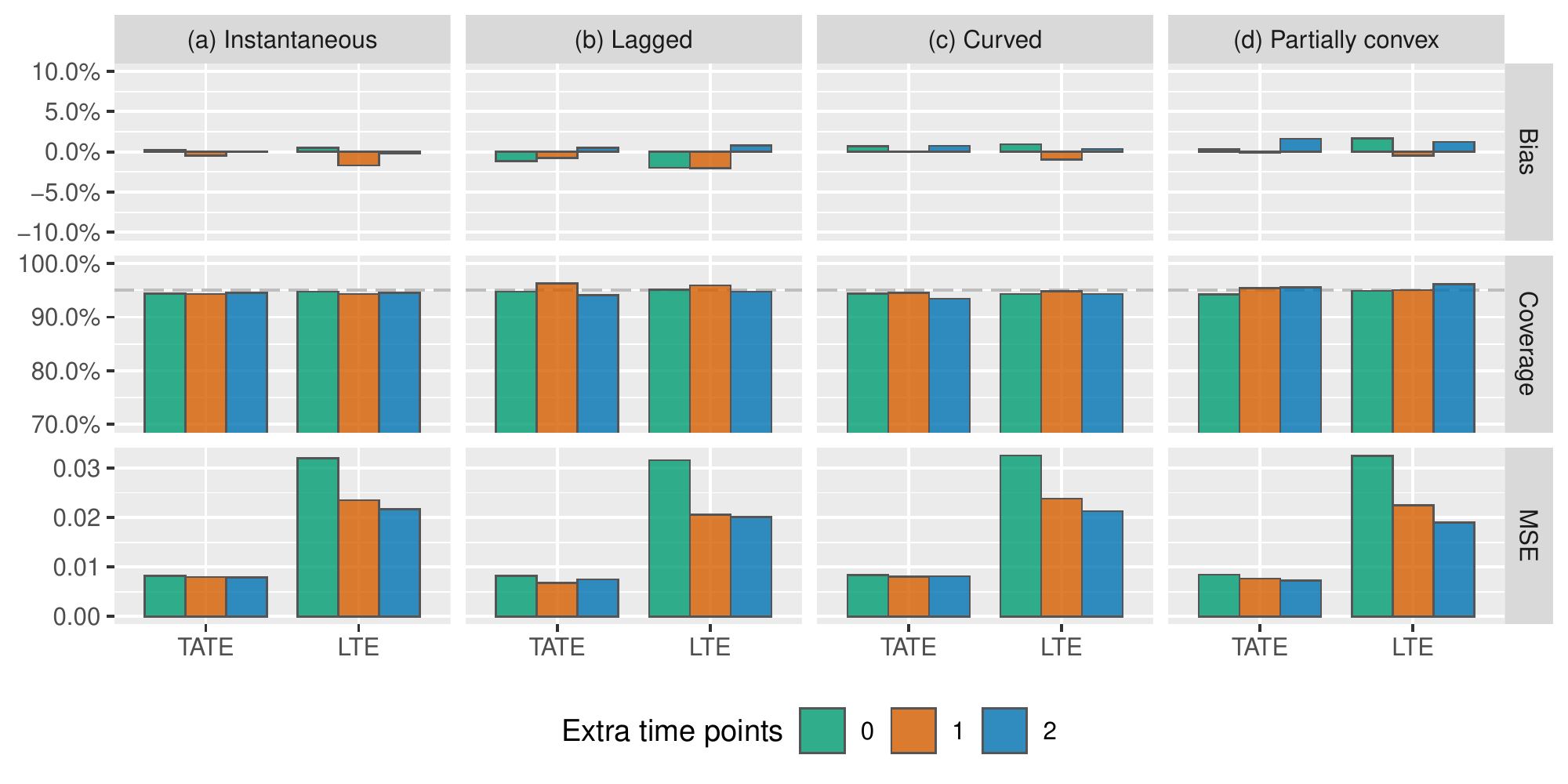}}
\caption{Simulation results: bias, coverage, and mean squared error (MSE) for the estimation of the TATE ($\Psi_{[0,J-1]}$) and LTE ($\Psi_{J-1}$) using the exposure time indicator (ETI) model with 0, 1, or 2 extra time points added to the end of the study}
\label{sim8}
\end{figure}

\section{Data analysis}

To illustrate the use of the models described in section \ref{sec_analysismodels}, we conducted secondary analyses of data from two different stepped wedge trials.\\

\subsection{Australia weekend services disinvestment trial}\label{real_disinv}

The first data analysis is of a stepped wedge trial examining the impact of the disinvestment (removal) of weekend health services from twelve hospital wards in Australia, previously analyzed by Haines et al.\cite{haines2017impact} Although the original investigators considered several outcomes, we focus on the the (log) length of hospital stay in days, treated as a continuous variable.\\

In our analysis, we fit six different models: immediate treatment (IT), exposure time indicator (ETI), exposure time indicator with a random treatment effect (ETI-RTE), restricted exposure time indicator with flattening time $s^*=3$ (RETI-3), natural cubic spline with four degrees of freedom (NCS-4), and monotone effect curve (MEC). For all models, study time is modeled as categorical. The MEC model uses the same prior that was used in simulations. Estimates of the TATE over the entire study ($\Psi_{[0,6]}$) and the LTE ($\Psi_{6}$) are given in Table \ref{table_real_disinv}, and estimates of the entire effect curve, along with pointwise confidence bands, are given in Figure \ref{real_disinv_tx}.\\

\begin{table}[h]
\centering
\begin{tabular}{ |c|c|c| } 
 \hline
 \textbf{Model} & \textbf{TATE (95\% CI)} & \textbf{LTE (95\% CI)} \\
 \hline
 IT & 0.11 (0.05--0.17) & 0.11 (0.05--0.17) \\
 ETI & 0.23 (0.13--0.33) & 0.29 (0.11--0.48) \\
 ETI-RTE & 0.27 (0.15--0.40) & 0.38 (0.15--0.61) \\
 RETI-3 & 0.19 (0.11--0.27) & 0.22 (0.13--0.31) \\
 NCS-4 & 0.22 (0.12--0.31) & 0.28 (0.10--0.45) \\
 MEC & 0.17 (0.09--0.26) & 0.23 (0.12--0.34) \\
 \hline
\end{tabular}
\caption{TATE and LTE estimates for Australia hospital}
\label{table_real_disinv}
\end{table}

\begin{figure}[H]
\centerline{\includegraphics[width=6in]{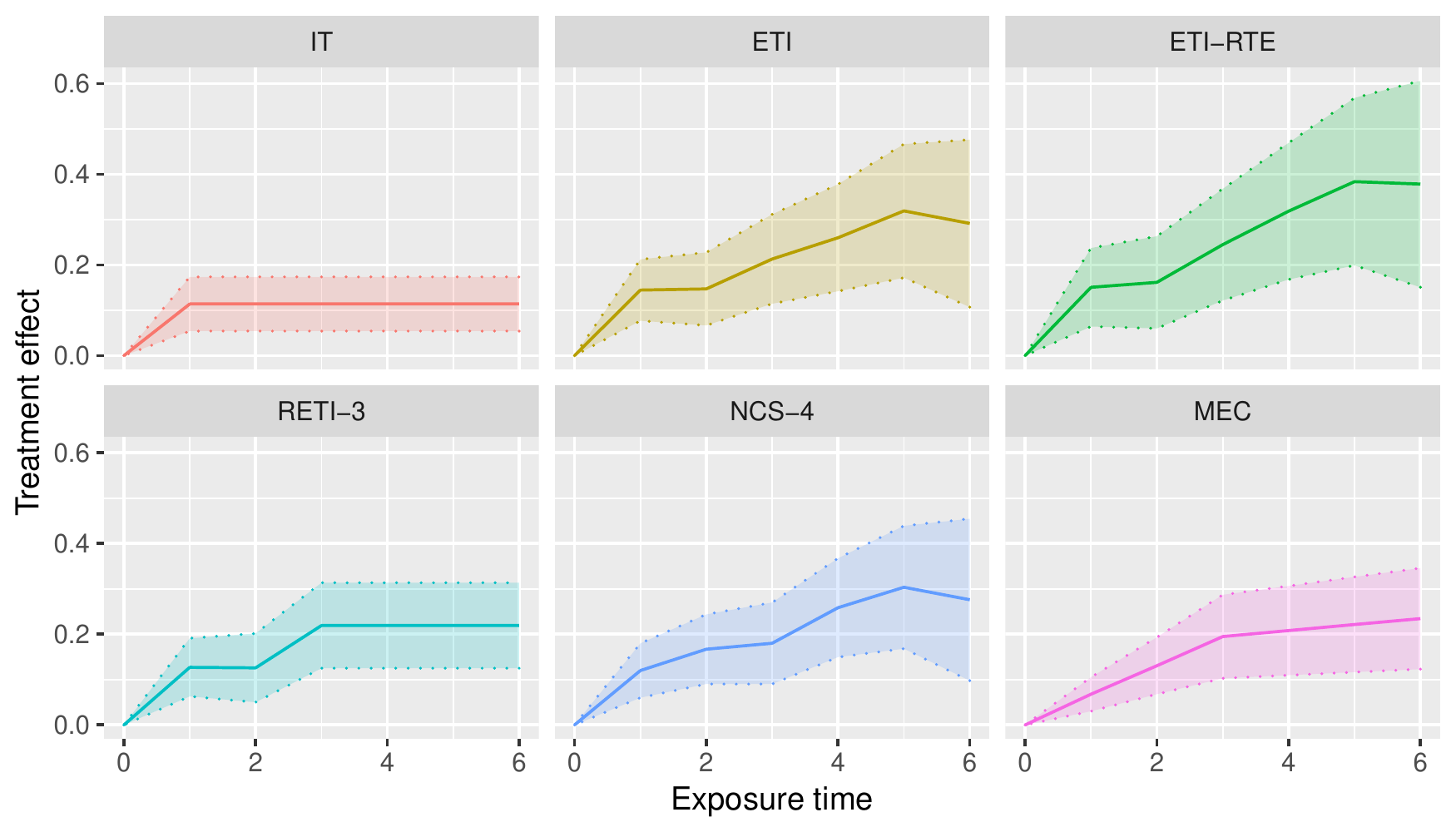}}
\caption{Estimation of the effect curve from the Australia disinvestment trial using the following six models: immediate treatment (IT), exposure time indicator (ETI), exposure time indicator with a random treatment effect (ETI-RTE), restricted exposure time indicator with flattening time $s^*=3$ (RETI-3), natural cubic spline with 4 degrees of freedom (NCS-4), monotone effect curve (MEC)}
\label{real_disinv_tx}
\end{figure}

Estimates of the TATE and LTE from the IT model are both much smaller than the corresponding estimates from all five other models. Additionally, the effect estimated from the IT model lies below the entire effect curve estimated from the ETI model. This is an example of the phenomenon described in section \ref{sec_behavior}. The effect curves estimated from the five models that account for a time-varying treatment effect are all qualitatively similar, which is reassuring. The confidence intervals resulting from the ETI model are the widest, which is as expected since this model makes no assumptions about the shape of the effect curve. The NCS model performs almost identically to the ETI model, which is not surprising for reasons discussed in section \ref{sec_sim23}. The ETI-RTE model yields higher TATE and LTE estimates than the ETI model, as well as wider confidence intervals; the wider CIs and nonzero variance component on the random treatment effect (not shown) suggest that the effect of treatment varies by cluster.\\

The RETI model was included for illustrative purposes even though we do not have a priori justification for the flattening assumption. As opposed to the ETI model, it leads to confidence intervals that are slightly narrower for TATE estimation and much narrower for LTE estimation, as expected.\\

The MEC model performs roughly similarly to the RETI model. As a sensitivity analysis, we also ran the MEC model with a (minimally informative) symmetric Dirichlet prior (i.e. $c_1=...=c_{J-1}=1$). This resulted in a similar TATE estimate of 0.18 (95\% CI: 0.09--0.28) but a much higher LTE estimate of 0.30 (95\%CI: 0.15--0.45), showing that the MEC model is indeed quite sensitive to prior selection.\\

\subsection{Washington State expedited partner treatment trial}\label{real_wa}

Next, we conducted a secondary analysis of data from the Washington State Community-Level Expedited Partner Treatment (EPT) Randomized Trial. \cite{golden2015uptake} This trial sought to test the effect of EPT, an intervention in which the sex partners of individuals with sexually transmitted infections are treated without medical evaluation, on rates of chlamydia and gonorrhea. In our analysis, we applied the same six models used for the first data analysis, but using a binomial GLM with a logit link and random intercepts for both cluster and site. Again, time was modeled as categorical. This study involved measurements at 15 study time points. The flattening point of $s^*=7$ chosen for the RETI model corresponds to the midpoint of the exposure time scale, as in the previous analysis, encoding the assumption that the effect curve flattens halfway through the trial. Estimates of the TATE over the entire study ($\Psi_{[0,14]}$) and the LTE ($\Psi_{14}$) are given in Table \ref{table_real_wa}, and estimates of the entire effect curve, along with pointwise confidence bands, are given in Figure \ref{real_wa_tx}.\\

\begin{table}[h]
\centering
\begin{tabular}{ |c|c|c| } 
 \hline
 \textbf{Model} & \textbf{TATE (95\% CI)} & \textbf{LTE (95\% CI)} \\
 \hline
 IT & 0.95 (0.86--1.04) & 0.95 (0.86--1.04) \\
 ETI & 1.17 (0.93--1.46) & 1.32 (0.55--3.18) \\
 ETI-RTE & 1.12 (0.90--1.41) & 1.25 (0.84--1.87) \\
 RETI-7 & 1.02 (0.87--1.18) & 1.05 (0.87--1.26) \\
 NCS-4 & 1.13 (0.93--1.39) & 1.35 (0.58--3.12) \\
 MEC & 1.01 (0.87--1.16) & 1.01 (0.82--1.24) \\
 \hline
\end{tabular}
\caption{TATE and LTE estimates for WA State EPT Trial (odds ratios)}
\label{table_real_wa}
\end{table}

\begin{figure}[H]
\centerline{\includegraphics[width=6in]{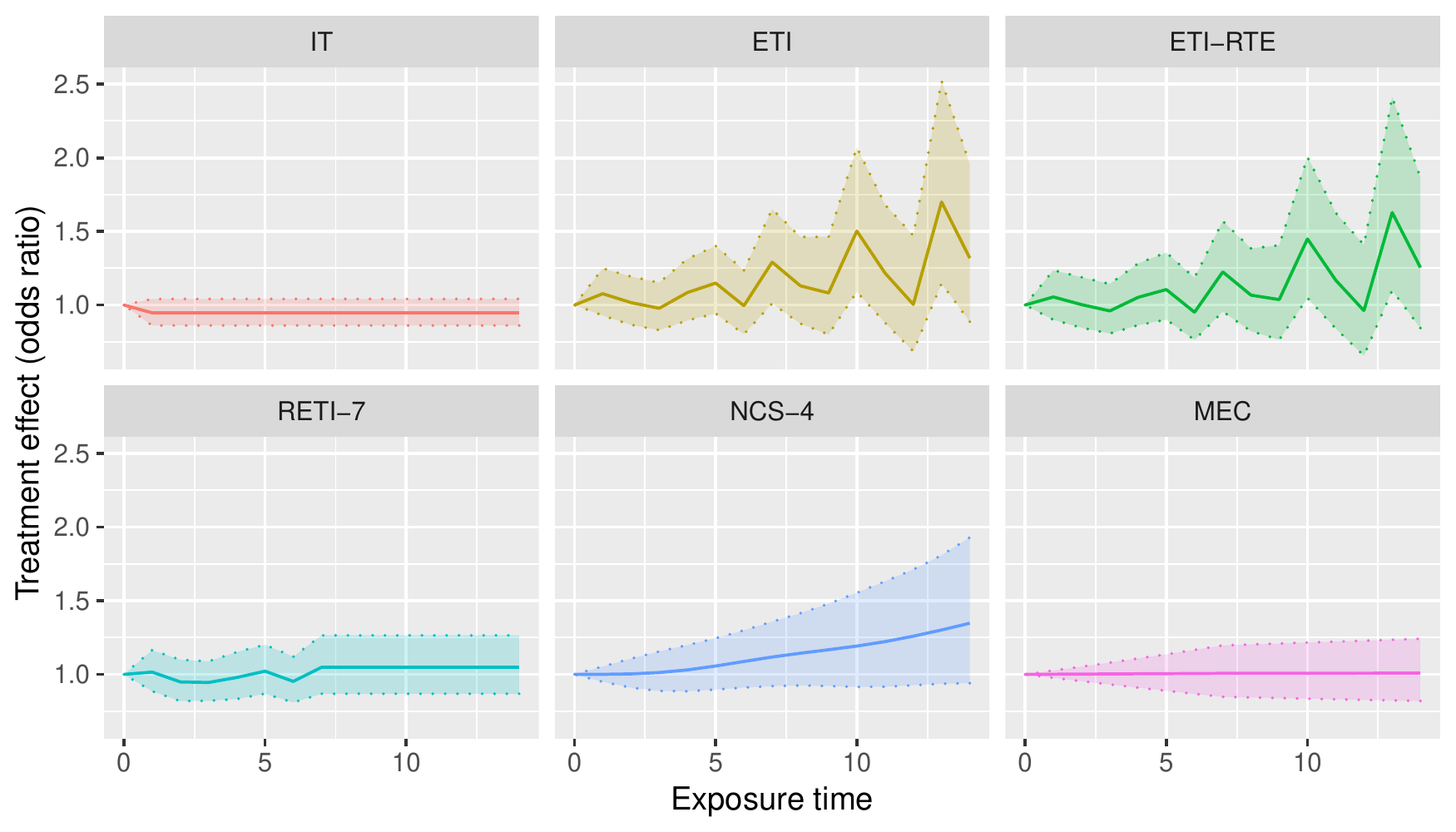}}
\caption{Estimation of the effect curve from the WA State EPT Trial using the following six models: immediate treatment (IT), exposure time indicator (ETI), exposure time indicator with a random treatment effect (ETI-RTE), restricted exposure time indicator with flattening time $s^*=7$ (RETI-7), natural cubic spline with 4 degrees of freedom (NCS-4), monotone effect curve (MEC)}
\label{real_wa_tx}
\end{figure}

In this analysis, we again observe that the effect estimated from the IT model lies entirely below the effect curve estimated from the ETI model. Furthermore, the trend is estimated in the opposite direction (odds ratio < 1) in the IT model relative to the models that account for a time-varying treatment effect (each odds ratio > 1). The IT estimate of the TATE is less than one, suggesting that the intervention has a small (nonsignificant) benefit, but the TATE estimates from the other models are all greater than one. However, all confidence intervals contain one (the null value) and so none of the associated Wald-type tests reject the null hypothesis of no effect, for both the TATE and the LTE.\\

The two ETI models (with and without a random treatment effect) yielded nearly identical estimation and inference, as the random treatment effect variance component was estimated to be close to zero (not shown), whereas these two models yielded different estimates in the previous dataset. We also see the smoothing behavior of the NCS model in action when comparing its estimated effect curve to that of the ETI model. As described earlier, this is expected, since with 14 time points there is a much greater difference in terms of the degrees of freedom between the ETI and NCS models. While one might expect the MEC model to yield an estimated effect curve similar to that of the NCS model, we do not observe this here; this is due to the effects of the prior. When we instead fit the MEC model with a symmetric Dirichlet prior (as we did as a sensitivity analysis for the previous dataset), we obtain a TATE estimate of 1.17 (95\% CI: 098--1.41) and an LTE estimate of 1.36 (95\%CI: 0.97--1.91), both of which are quite different than the results in Table \ref{table_real_wa}. These TATE and LTE estimates, as well as the estimate of the entire effect curve, are similar to those of the NCS model. Again, this demonstrates the sensitivity of the MEC model to the priors.\\

\section{Discussion}\label{sec_discussion}

In this paper, we characterize the risks of not accounting for a time-varying treatment effect in the context of a stepped wedge design, introduce terminology for describing estimands of interest, introduce several models that can be used for estimation and hypothesis testing, evaluate the operating characteristics of these models via simulation, and demonstrate their use on two real datasets.\\

\subsection{The immediate treatment (IT) model and estimands of interest}

If a time-varying treatment effect is present and the IT model is used, estimates of both the TATE and the LTE can be severely biased, confidence interval coverage can be unacceptably low, and MSE can be very high. As noted in section \ref{sec_behavior}, the expectation of the TATE/LTE estimate can be represented as a weighted sum of the individual point treatment effects $\delta(1),...,\delta(J-1)$, but with some negative weights. This implies that under certain effect curves, TATE and LTE estimates resulting from the use of the IT model can converge to a value of the opposite sign as the true parameters. This result is quite surprising, and although similar phenomena have been examined in the context of two-way fixed effects designs\cite{de2016double,callaway2020difference}, this has never previously been discussed in the stepped wedge literature. Because the impact of this form of model misspecification on estimation and inference is so severe, we recommend that the IT model should not be used moving forward for confirmatory analyses unless the assumption of an immediate treatment effect is justifiable based on contextual knowledge of the intervention. Even when the assumption is justifiable, a sensitivity analysis should be conducted using a model that allows for a time-varying treatment effect. Furthermore, it may be worth reanalyzing data from past stepped wedge trials in which the immediate treatment effect assumption is questionable.\\

In the case of an immediate treatment effect, the effect curve is constant over time, and so the phrase ``treatment effect'' is well-defined. However, when the effect curve changes over time, this phrase is ambiguous and researchers must take care to specify precisely what they are estimating. We defined the time-averaged treatment effect (TATE), the point treatment effect (PTE), and the long-term treatment effect (LTE) to distinguish between the different estimands a researcher may be interested in. In the case of the TATE, an alternative definition we considered was $(s_2-s_1)^{-1}\sum_{s=s_1+1}^{s_2}\delta(s)$, the average of the true point treatment effects at the study measurement times. We chose to not use this definition for two main reasons. First, doing so would imply that the estimand of interest is entangled with the study design, in the sense that the quantity and meaning of the $\delta(s)$ terms directly depends on the number of time steps in the study and the spacing between the steps. Second, since time is continuous in reality, we believe that an integral-based average is conceptually closer to what researchers will typically be interested in. However, the estimands may depend on the study design in practice, in the sense that possible choices for $(s_1,s_2)$ in $\Psi_{[s_1,s_2]}$ or $s_0$ in $\Psi_{s_0}$ are limited by the number and length of time steps in the study design. Similarly, the effect curve flattening assumption necessary for the LTE to be well-defined requires the study to be of a sufficient length to achieve flattening.\\

\subsection{Features unique to the setting of time-varying treatment effects}

For the models we consider that account for a time-varying treatment effect (other than the RETI model), it will generally be the case that the variance of the PTE increases with exposure time. Intuitively, this occurs because if there are $Q$ sequences in total, all $Q$ will be observed at exposure time 1, but only $Q-1$ will be observed at exposure time 2, only $Q-2$ will be observed at exposure time 3, and so on. However, this behavior does not apply to the RETI model (with flattening point $s^*$). With this model, although the variance of each PTE estimate $\Psi_{s_0}$ for $s_0<s^*$ will increase with increasing exposure time, the variance of the LTE estimate (which is equivalent to the PTE for any time $s_0 \ge s^*$) will generally be smaller, since it pools information from all exposure times greater than $s^*$. Thus, estimation of both $\Psi_{[s_1,s_2]}$ for $s_1>s^*$ and $\Psi_{s_0}$ for $s_0>s^*$ can be done with greater precision, as we saw in the simulation study.\\

We also note that some of the existing guidance for the design of stepped wedge trials requires modification when there is a time-varying treatment effect. Normally, data collection stops after all clusters are observed in the treatment state. However, in our simulations, we explored the potential for collecting data at additional time points and showed that this can lead to increased precision, particularly if the LTE is of interest. Additionally, care should be taken when considering designs involving different time step lengths. For example, consider a design in which all clusters are measured at baseline, at three months, and at eight months (in terms of study time). In this case, the first three measurements for sequence 1 will occur at exposure times 0, 3, and 8, whereas the first three measurements for sequence 2 will occur at exposure times 0, 0, and 5. This is not an issue for the IT model, but an ETI model will require more exposure time parameters (i.e. the $\delta_s$ terms), since each of these parameters represents the treatment effect at a specific exposure time. The NCS model is very appealing in this scenario, since it allows for the researcher to explicitly limit the number of model parameters; that is, the number of spline basis functions can remain constant even with a growing number of unique exposure time values.\\

\subsection{Which model to choose}

The ETI model is saturated with respect to exposure time, and thus makes no assumptions about the shape of the effect curve (assuming measurements are taken at regular intervals); in this sense, it is the most robust model that we consider. If the study time trend is correctly specified and a linear link function is used, the ETI model will always yield unbiased estimates of the PTE and the LTE, as well as (barring the error due to the use of a Riemann sum in approximating the effect curve) the TATE for any time interval. If instead a nonlinear link function is used, estimates will be unbiased if the heterogeneity model component is also correctly specified. However, the ETI model will also lead to estimators with the greatest variance. It will often be impractical to estimate the LTE using this model, since the information about the change in the effect curve between the second-to-last and last time points comes from just a single sequence. This is a problem not only in terms of variance, but also in terms of generalizability, especially if there are only a few clusters per sequence. This being said, we see that the LTE estimate in the Australia disinvestment trial is statistically significant for the ETI model (as well as for the other five models considered). We anticipate that the ETI model will be a practical choice in many applications for TATE estimation.\\

The RETI model, which assumes that the effect curve flattens out by exposure time $s^*$, is an attractive option when this assumption can be reliably made. When this assumption is correct, the RETI model leads to TATE and LTE estimators with considerable lower MSE. Furthermore, lower values of $s^*$ will generally lead to estimators with greater precision. Also, as noted in section \ref{sec_analysismodels}, the IT model can be seen as the extreme case of the RETI model in which $s^*=1$, and so the RETI model represents something of a compromise betwen the ETI model and the IT model. However, as with the IT model, violations of the flattening assumption can lead to substantial bias and undercoverage. If the RETI model is used, the analyst may choose to fit an ETI model as a sensitivity analysis.\\

With the NCS model, although we observe the smoothing effect in the WA State data analysis, we don't see evidence of improved estimation or inference in our simulation study. However, we still view this model as useful, particularly when the entire effect curve is of interest. Further research is needed to determine optimal knot number and placement. Also, as mentioned above, the NCS model is well-suited to trials in which the lengths of the time steps differ, since otherwise the number of exposure time parameters (i.e. the $\delta_s$ terms) would be very high, yielding TATE estimators with high variance.\\

We showed that the use of a random treatment effect is beneficial in the ETI model, but the same conclusion applies to the RETI, NCS, and MEC models. There is little penalty for including a random treatment effect if it is feasible to do so, since performance does not suffer when no random treatment effect is present in the data, but if it is present in the data we see a considerable performance improvement, particularly in terms of coverage.\\

Next, we discuss the MEC model. This model performs well in simulations for estimation of both the TATE and the LTE. However, as we demonstrated in both the simulations and the real data analyses, it is highly sensitive to the choice of prior. We observed that for both real data analyses, if we instead used a symmetric Dirichlet prior (i.e. $c_1=...=c_{J-1}=1$), our estimates shifted considerably. If we run the same set of simulations summarized \ref{sec_sim23} but with a symmetric Dirichlet prior instead of a $\text{Dirichlet}(5\omega,5\omega,5\omega,\omega,\omega,\omega)$ prior, we obtain estimates with greater bias, lower coverage, and higher MSE (see Appendix \ref{app_additional_sims}). As mentioned in section \ref{sec_analysismodels}, estimates obtained by this model are influenced not just by the monotonicity constraint, but also by the informative prior. In general, we see in both simulated and real data that the prior has a considerable influence on the resulting estimates, and as such, we recommend that this model is only used for exploratory analyses.\\

\subsection{Assumptions and limitations}

A critical assumption made throughout this work is that the treatment effect varies as a function of exposure time, but \textit{not} as a function of study time. While this assumption will hold in many scientific settings, researchers should carefully consider whether it holds in each particular study. A notable example of when this assumption would fail is when there is a major external shock (e.g. a natural disaster) that occurs in the middle of the study period that affects implementation of the intervention. Also, while all of the analysis models we consider use random effects to account for the correlation structure of the data, other approaches are possible.\cite{scott2017finite,kennedy2020novel,thompson2018robust,hughes2020robust} Most of these alternative approaches, such as GEE, differ mainly in terms of how the dependence between observations is handled. Thus, by making analogous modifications to the fixed effects structures, these approaches can be adapted to settings in which there is a time-varying treatment effect. However, we did not study these approaches in depth. Another limitation of this work is the set of simulation scenarios we considered was limited, and represents a small fraction of the set of possible data-generating mechanisms, in terms of true effect curves, ICC values, outcome variable types, covariance structures, etc.\\

\subsection{Future research directions}

We see a number of open questions that could be explored through future methodological work. First, the use of a time-varying treatment effect model will have a substantial impact on statistical power. If one makes assumptions about the shape of the effect curve, it should be straightforward to develop an equation to calculate power, at least for the ETI and RETI models. Second, the setting of a time-varying treatment effect may lead to additional implications on study design other than the possibility of adding additional time points. These implications will likely differ depending on the estimand of interest. The topic of model misspecification is also worthy of further exploration. In section \ref{sec_behavior}, we explored the behavior of the IT model under a time-varying treatment effect but focused on the Hussey and Hughes model; it would be informative to perform similar analyses for other other models (e.g. the behavior of the RETI model when the flattening assumption does not hold), correlation structures, and sampling schemes (e..g. cohort sampling instead of cross-sectional sampling). It would similarly be useful to examine how the models we study here perform when the treatment effect varies as a function of study time rather than exposure time, as well as to construct models that handle this scenario. Finally, it would be worth doing further research into alternative ways to estimate the time trend in the context of a time-varying treatment effect. Since moving from the IT model to the ETI model results in estimators with higher variance, imposing a more restrictive model for the time trend could potentially help to gain back some precision. For example, it would be straightforward to use a natural cubic spline (or other smoothing estimator) to estimate the time trend. Other approaches that have been considered include the use of linear\cite{hemming2017analysis} and quadratic\cite{nickless2018mixed} time trends, or the exclusion of time trends entirely for short trials.\cite{zhou2017cross} It would be useful to understand what precision gains are possible, as well as the costs associated with time trend misspecification.\\

\section{Conclusions}

If a stepped wedge trial is testing a treatment that varies as a function of exposure time, the use of a model that assumes an immediate treatment effect can lead to serious errors in both estimation and inference. We introduced several models that account for time-varying treatment effects and evaluated their operating characteristics via simulation. We recommend that a model that accounts for time-varying treatment effects is always used in stepped wedge trials moving forward unless the researcher has compelling evidence that the treatment effect is immediate.\\

%\backmatter

\section*{Acknowledgments}

This publication was supported by the National Center For Advancing Translational Sciences of the National Institutes of Health under Award Number UL1 TR002319, as well as the National Institutes of Health under award number AI29168. The content is solely the responsibility of the authors and does not necessarily represent the official views of the National Institutes of Health.\\

\subsection*{Financial disclosure}

None reported.

\subsection*{Conflict of interest}

The authors declare no potential conflict of interests.\\

\section*{Supporting information}

A mock stepped wedge trial dataset and an R code file to run the models described in this paper are included as supplementary files. Additionally, all of the code for the simulation study can be found online at https://github.com/Avi-Kenny/Stepped.wedge.delay.\\

\subsection*{Data availability statement}

The data that support the findings of this study are available from the corresponding author upon reasonable request.\\

\appendix

\section{Behavior of the IT model under alternative correlation structures}\label{app_corr}

In section \ref{sec_behavior}, we investigated the behavior of the treatment effect estimator resulting from the IT model (specifically the Hussey and Hughes model), and showed that the expectation of this estimator can be expressed as a weighted sum of the individual point treatment effects $\{\delta(1),...,\delta(J-1)\}$. In this section, we provide additional results illustrating the behavior of these weights for several correlation structures. This analysis shows that the behavior of the IT model treatment effect estimator observed in section \ref{sec_behavior} is not restricted to the model in which heterogeneity is modeled as a single cluster-level random intercept.\\

\subsection{Exchangeable correlation structure}

First, we use equation (\ref{eq_weightfunction}) to show how the weights vary as a function of $\phi$ for select values of $Q$ under the assumption of an exchangeable correlation matrix, the structure implied by model (\ref{eq_husseyhughes}). Results are given in Figure \ref{weights}.

\begin{figure}[H]
\centerline{\includegraphics[width=6in]{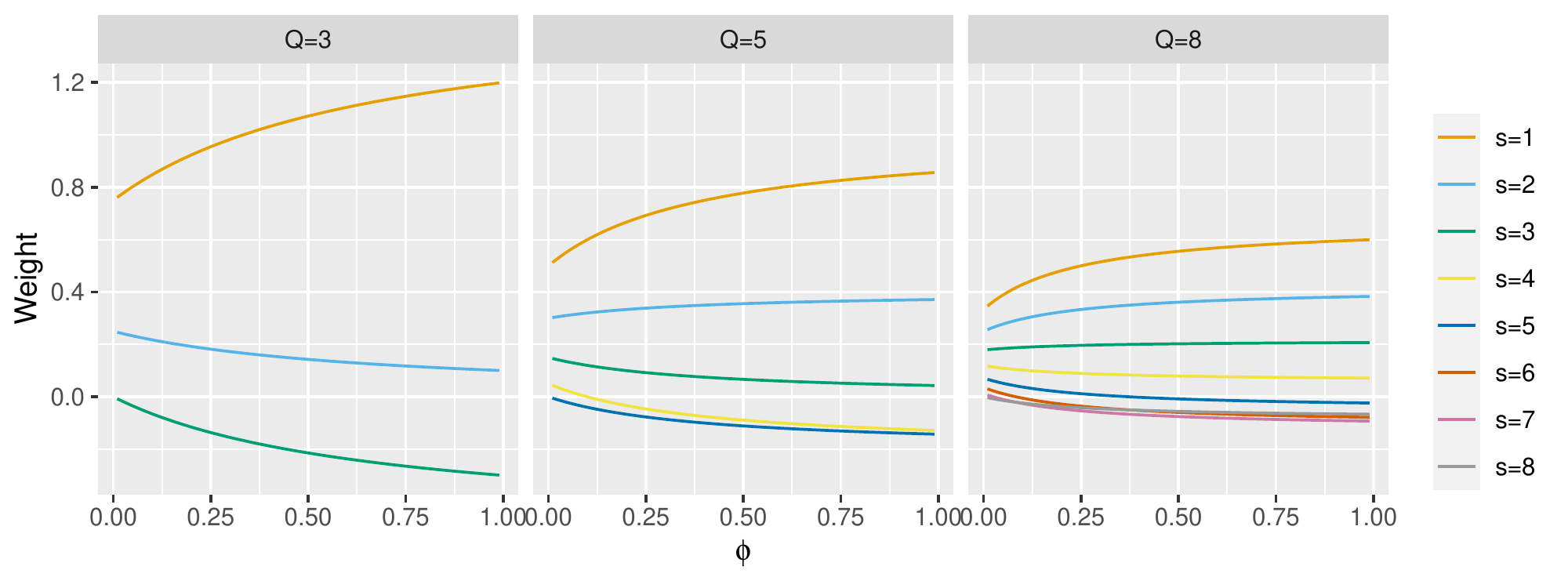}}
\caption{Weights $w(Q,\phi,s)$ plotted as a function of $\phi$ for $Q \in \{3,5,8\}$.\label{weights}}
\end{figure}

We observe that the weight $w(Q,\phi,1)$, which corresponds to the point treatment effect after a single time step, receives the most weight, regardless of the value of $Q$ or $\phi$. Additionally, in the $Q=8$ case, we see that half of the weights are negative as $\phi \rightarrow 1$. This implies that, under certain true effect curves, the expected value of the IT model treatment effect estimator may be of the opposite sign as the true TATE or LTE. As noted in section \ref{sec_behavior}, the value of $\phi$ can be quite large if the number of individuals per cluster is large, even if $\rho$ is small, and so we expect a broad range of $\phi$ values to be encountered in real trials.

\subsection{Correlation structure implied by a random time effect}

Next, we produce analogs of Figure \ref{weights} for several alternative correlation structures. Rather than deriving analytical results, we use Mathematica\cite{Mathematica} software to calculate treatment effects estimates as a function of $\{\delta(1),...,\delta(J-1)\}$ by computing the weighted least squares estimator. We consider a design involving three sequences with one cluster each and four time points.\\

We start with the correlation structure implied by a random time effect. This model is the same as the Hussey and Hughes model, but with the addition of a Normal random effect indexed by both cluster and time representing a unique cluster-by-time interaction. The resulting correlation matrix involves two parameters; $\rho_w$ represents the correlation between two observations within a cluster measured at the same time point and $\rho_b$ represents the correlation between two observations within a cluster measured at different time points. This structure is depicted in (\ref{eq_corr_random_time}), adapted from Li et al.\cite{li2020mixed}, for a trial involving three time points and two individuals per time point.

\begin{equation}\label{eq_corr_random_time}
	\left(
		\begin{array}{cc|cc|cc}
		1 & \rho_w & \rho_b & \rho_b & \rho_b & \rho_b \\
		\rho_w & 1 & \rho_b & \rho_b & \rho_b & \rho_b \\
		\hline
		\rho_b & \rho_b & 1 & \rho_w & \rho_b & \rho_b \\
		\rho_b & \rho_b & \rho_w & 1 & \rho_b & \rho_b \\
		\hline
		\rho_b & \rho_b & \rho_b & \rho_b & 1 & \rho_w \\
		\rho_b & \rho_b & \rho_b & \rho_b & \rho_w & 1
		\end{array}
	\right)
\end{equation}\\

The resulting weights are depicting in Figure \ref{weights2} for a design involving two individuals per cluster, for select values of $\rho_w$ and $\rho_b$.\\

\begin{figure}[H]
\centerline{\includegraphics[width=6in]{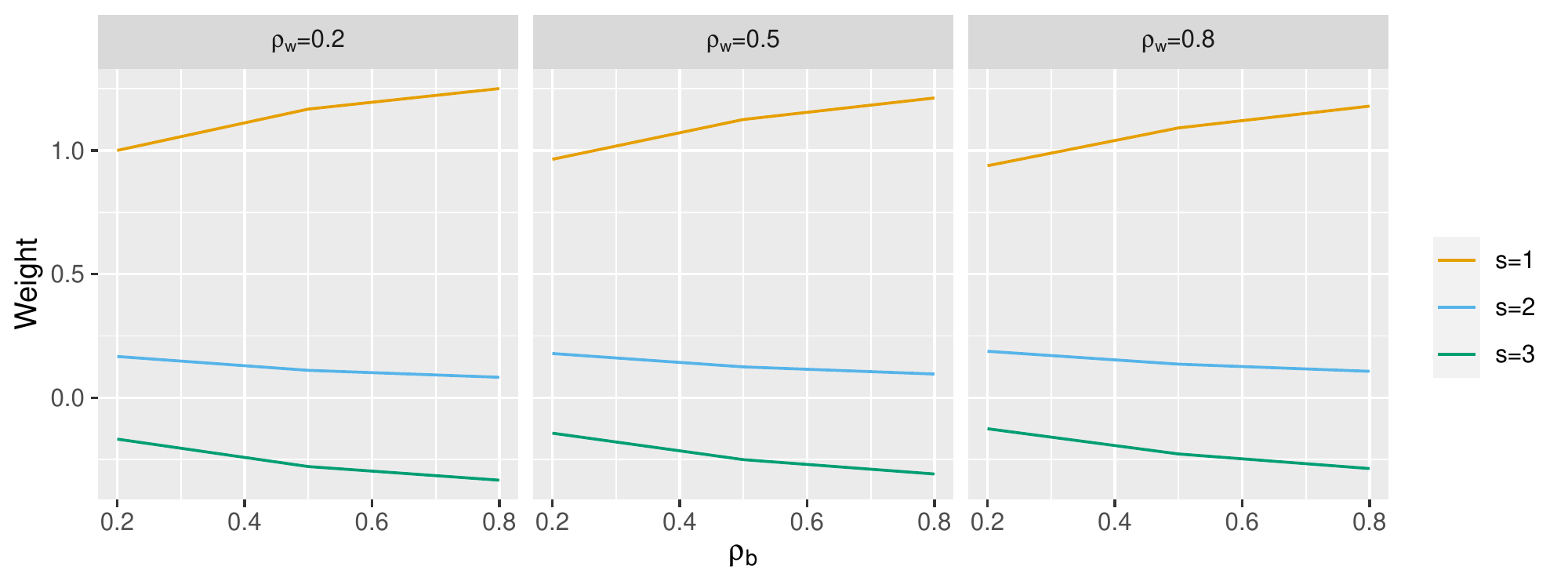}}
\caption{Weights plotted for select values of $\rho_w$ and $\rho_b$ for $Q=3$.\label{weights2}}
\end{figure}

\subsection{Correlation structure implied by a random treatment effect}

Next, we computed the weights for the correlation structure implied by a random treatment effect (we discussed random treatment effects in section \ref{sec_rte}). The resulting correlation matrix involves three parameters; $\rho_0$ represents the correlation between two observations in the control state, $\rho_1$ represents the correlation between two observations in the treatment state, and $\rho_10$ represents the correlation between two observations, one of which is in the control state and one of which is in the treatment state. The resulting weights are depicting in Figure \ref{weights3} for select values of $\rho_0$, $\rho_1$, and $\rho_{10}$.\\

\begin{figure}[H]
\centerline{\includegraphics[width=6in]{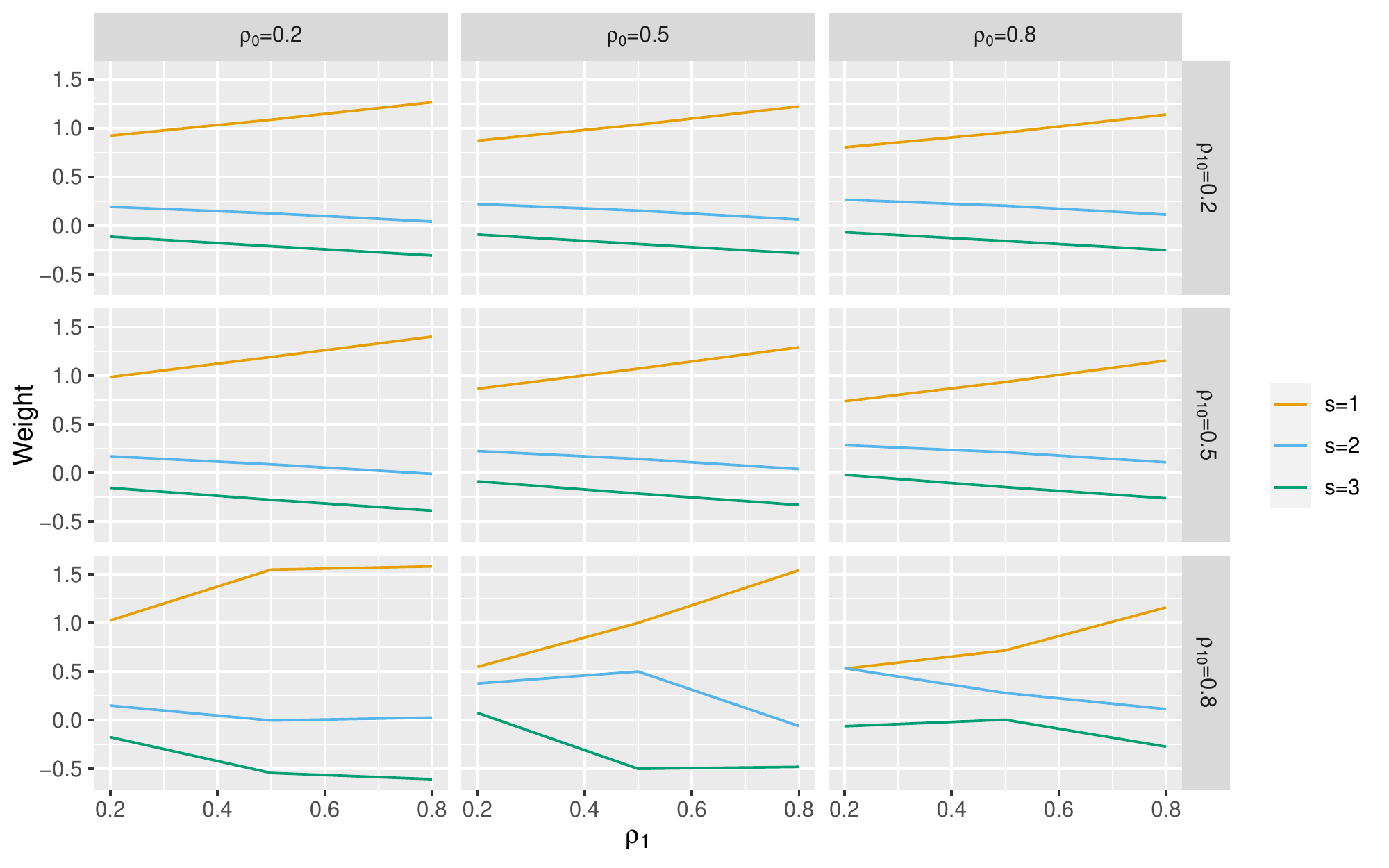}}
\caption{Weights plotted for select values of $\rho_0$, $\rho_1$, and $\rho_10$ for $Q=3$.\label{weights3}}
\end{figure}

We see that a similar pattern of weights holds for all three correlation structures considered. This indicates that the problem of IT model misspecification discussed in section \ref{sec_behavior} is not restricted to the Hussey and Hughes model, but is a general problem with models that assume an immediate treatment effect.\\

\section{Additional simulation results}\label{app_additional_sims}

\subsection{Performance of models in the presence of random treatment effects}\label{sec_sim_rte}

In the set of simulations summarized in Figure \ref{sim7}, data were generated according to the same model used for the simulations in section \ref{sec_sim23}, except additionally with a cluster-level random treatment effect generated according to (\ref{eq_eti2re}) with $\nu=1$ and $\rho=-0.2$. Data were analyzed using an ETI model that does not account for random treatment effects, and a second ETI model that implements the random treatment effect structure given in (\ref{eq_eti2re}); the latter is abbreviated as ``ETI-RTE''. The results are consistent with findings from previous studies on variance structures in the context of stepped wedge designs. When a random treatment effect is present in the data but we fail to account for it, we suffer considerably in terms of coverage and MSE.

\begin{figure}[H]
\centerline{\includegraphics[width=6in]{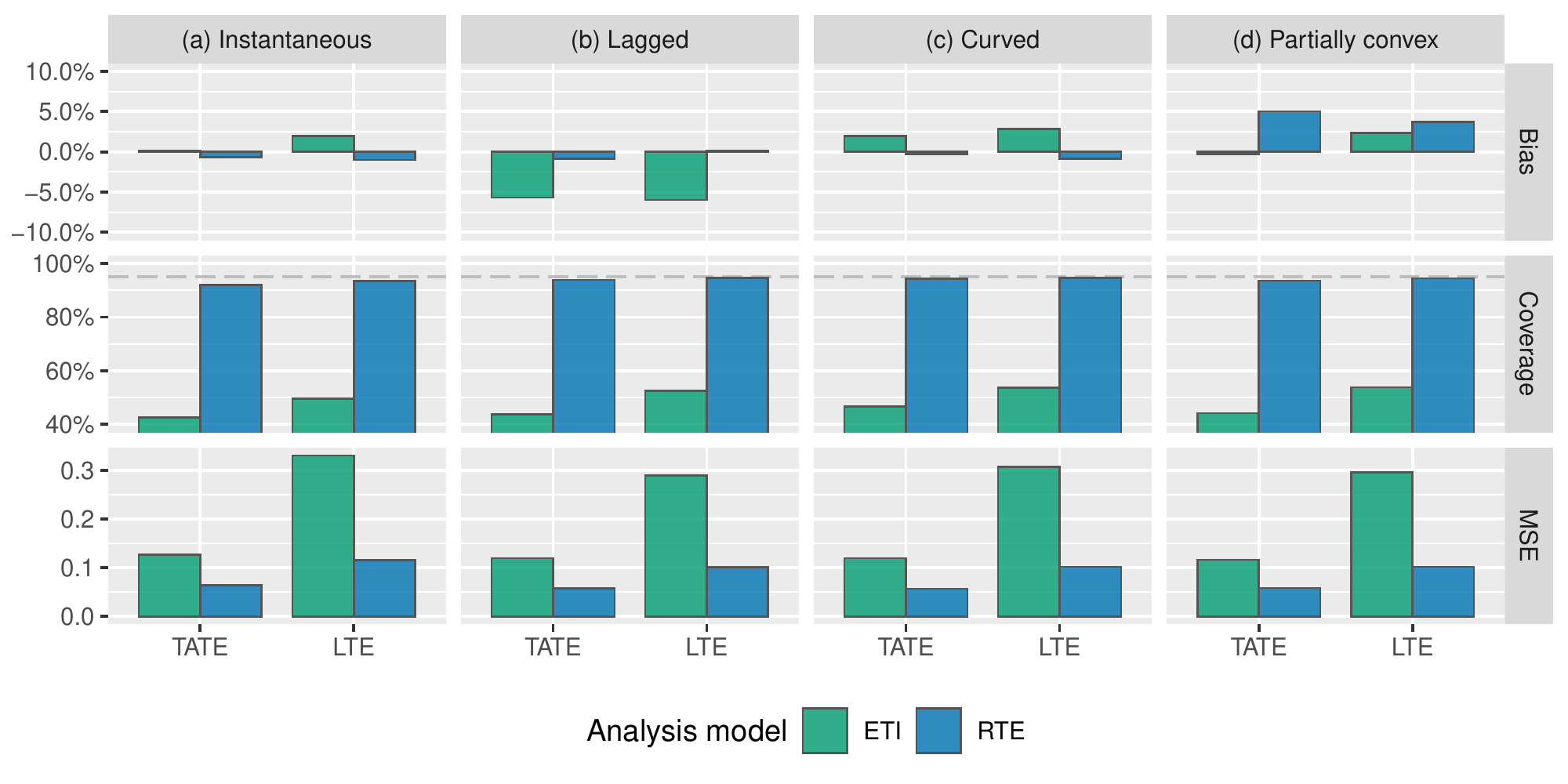}}
\caption{Simulation results: bias, coverage, and mean squared error (MSE) for the estimation of the TATE ($\Psi_{[0,J-1]}$) and LTE ($\Psi_{J-1}$) when data are generated with random treatment effects, using the exposure time indicator model (ETI) and the exposure time indicator model with a random treatment effect (ETI-RTE)}
\label{sim7}
\end{figure}

We also performed an analogous set of simulations in which data were instead generated without a random treatment effect (i.e. with $\nu=0$). In these simulations, the ETI and ETI-RTE models performed similarly across all three metrics; results are shown in Figure \ref{sim6}.\\

\begin{figure}[H]
\centerline{\includegraphics[width=6in]{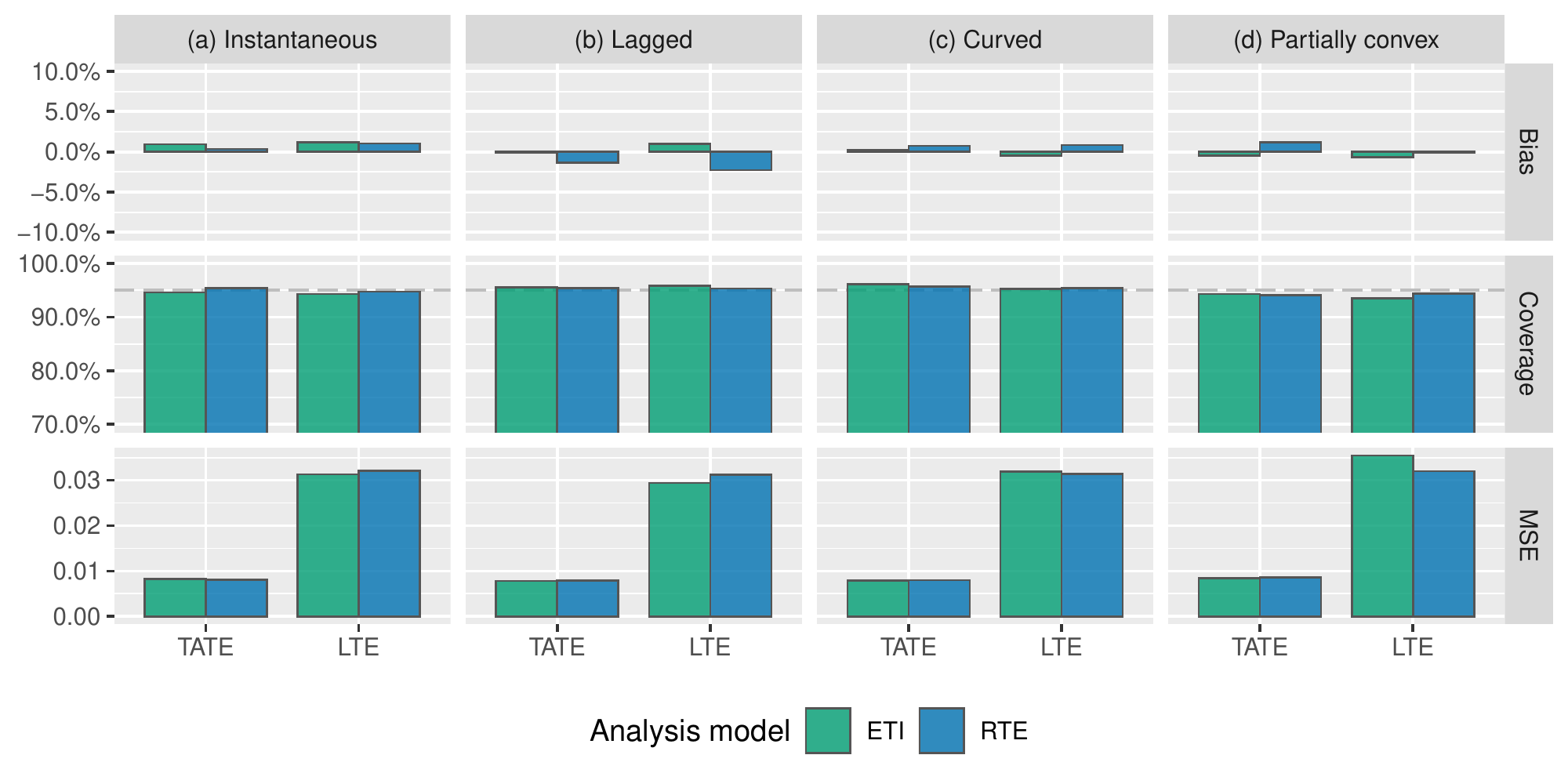}}
\caption{Simulation results: bias, coverage, and mean squared error (MSE) for the estimation of the TATE ($\Psi_{[0,J-1]}$) and LTE ($\Psi_{J-1}$) when data are generated without random treatment effects, using the exposure time indicator model (ETI) and the exposure time indicator model with a random treatment effect (ETI-RTE)}
\label{sim6}
\end{figure}

Together, these results suggest that random treatment effect terms should be used whenever possible, since very little is lost in terms of inferential performance when no random treatment effect is present in the data, whereas if it is present in the data we achieve much better performance.\\

\subsection{Use of a symmetric Dirichlet prior for the MEC model}

As noted in the discussion, if we use a symmetric Dirichlet prior (i.e. $\text{Dirichlet}(\omega,\omega,\omega,\omega,\omega,\omega)$) in simulations instead of a $\text{Dirichlet}(5\omega,5\omega,5\omega,\omega,\omega,\omega)$ prior, we see much different behavior of the MEC model. Estimates are more biased and undercoverage is generally worse. We also don't see gains in MSE that are as dramatic, and in some cases, MSE is worse in the MEC model than it is in the ETI model. Results are given in Figure \ref{sim2_symmetric_dir}.\\

\begin{figure}[H]
\centerline{\includegraphics[width=6in]{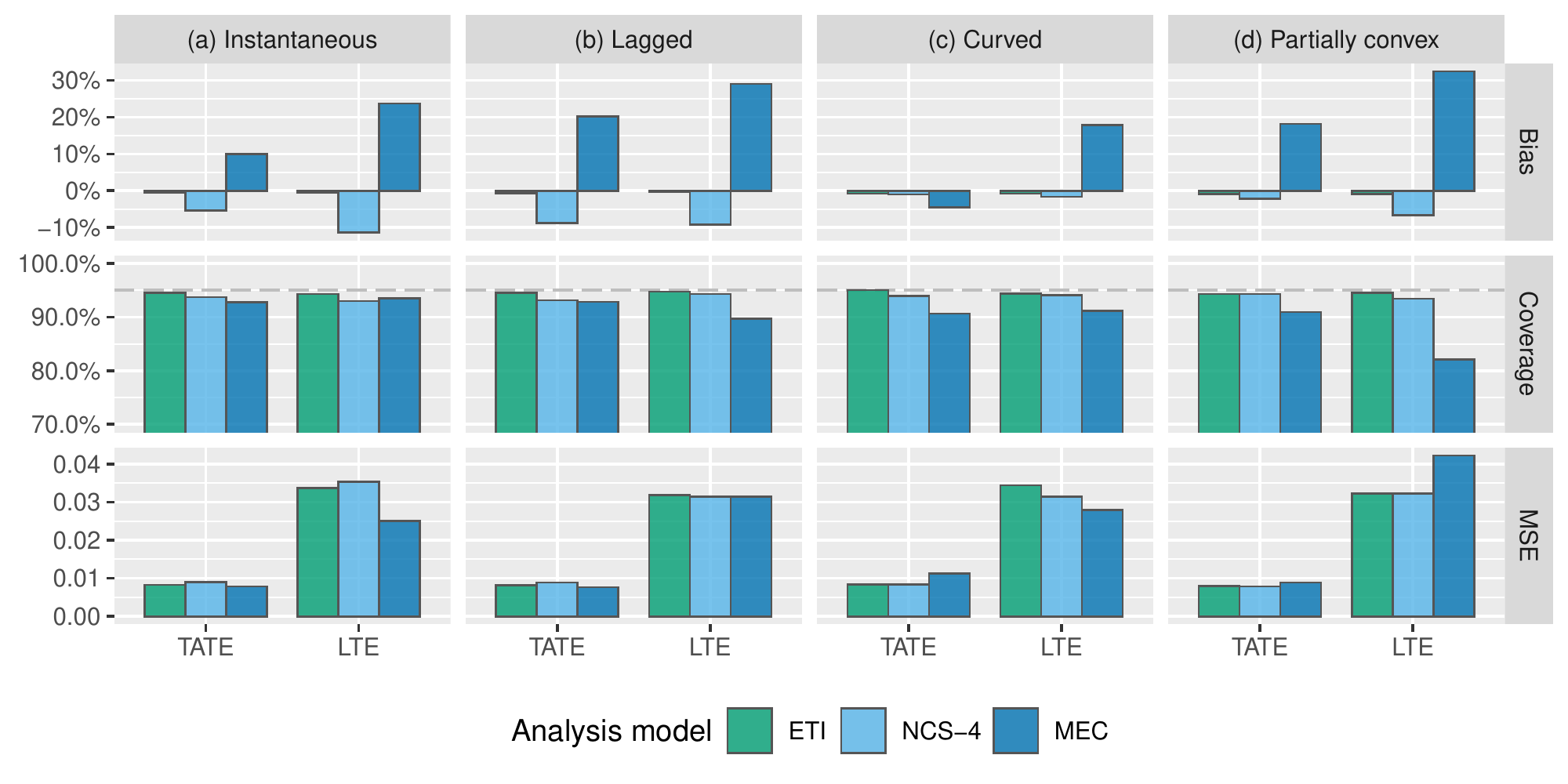}}
\caption{Simulation results: bias, coverage, and mean squared error (MSE) for the estimation of the TATE ($\Psi_{[0,J-1]}$) and LTE ($\Psi_{J-1}$) using the following four models: exposure time indicator (ETI), natural cubic spline with 4 degrees of freedom (NCS-4), monotone effect curve with a symmetric Dirichlet prior (MEC)}
\label{sim2_symmetric_dir}
\end{figure}

\section{Use of trapezoidal TATE estimators}\label{app_trapezoid}

As noted in section \ref{sec_analysismodels}, the time-averaged treatment effect (TATE) is an integral which must be approximated using measurements taken in discrete time. Throughout the paper, we chose to use TATE estimators based on right-hand Riemann sums, due to their simplicity and to facilitate more natural performance comparisons between different effect curves. However, it is also worth considering estimators based on trapezoidal Riemann sums. For the ETI model, the right-hand Riemann sum estimator was given by the following:

\begin{equation}\label{eq_est1_app}
	\hat{\Psi}_{[s_1,s_2]} \equiv \frac{1}{s_2-s_1} \sum_{r=1}^{s_2-s_1} \hat{\delta}_{s_1+r}
\end{equation}\\

The corresponding estimator based on a trapezoidal Riemann sum is given in (\ref{eq_est1b}) below, where for notational convenience we set $\hat{\delta}_0 \equiv 0$:

\begin{equation}\label{eq_est1b}
	\hat{\Psi}_{[s_1,s_2]}^* \equiv \frac{1}{s_2-s_1} \sum_{r=1}^{s_2-s_1} \frac{\hat{\delta}_{s_1+r-1}+\hat{\delta}_{s_1+r}}{2}
\end{equation}\\

It is informative to compare the variances of these two estimators for the case in which $s_1=0$. We can write the following:

\begin{equation}\label{eq_est1_variance}
\begin{aligned}
	\text{Var} \left( \hat{\Psi}_{[0,s_2]} \right) &=
	\frac{1}{s_2^2} \text{Var} \left( \sum_{r=1}^{s_2-1} \hat{\delta}_r + \hat{\delta}_{s_2} \right) &=
	\frac{1}{s_2^2} \left[
		\text{Var} \left( \sum_{r=1}^{s_2-1} \hat{\delta}_r \right) +
		\text{Var} \left( \hat{\delta}_{s_2} \right) +
		\text{2 Cov} \left( \sum_{r=1}^{s_2-1} \hat{\delta}_r , \hat{\delta}_{s_2} \right)
	\right] \\
	\text{Var} \left( \hat{\Psi}_{[0,s_2]}^* \right) &=
	\frac{1}{s_2^2} \text{Var} \left( \sum_{r=1}^{s_2-1} \hat{\delta}_r + \frac{1}{2} \hat{\delta}_{s_2} \right) &=
	\frac{1}{s_2^2} \left[
		\text{Var} \left( \sum_{r=1}^{s_2-1} \hat{\delta}_r \right) +
		\frac{1}{4} \text{Var} \left( \hat{\delta}_{s_2} \right) +
		\text{Cov} \left( \sum_{r=1}^{s_2-1} \hat{\delta}_r , \hat{\delta}_{s_2} \right)
	\right]
\end{aligned}
\end{equation}\\

Simulation results suggest that the covariance term in (\ref{eq_est1_variance}) is often positive, and so it is reasonable to expect the variance of the trapezoidal estimator to be greater than the variance of the right-hand Riemann estimator. This intuitively makes sense, since the two differ only by a factor of $\frac{1}{2}\hat{\delta}_{s_2}$, and $\hat{\delta}_{s_2}$ will have higher variance than any other $\hat{\delta}_t$ term for $t<s_2$. Therefore, it is worth considering use of the trapezoidal estimators if it is reasonable to expect the effect curve to be somewhat smooth. Although we focused here on the ETI model, similar trapezoidal estimators can be formed for the other models we considered in this paper.\\

\section{Proofs}\label{app_proof}

Suppose we have a standard stepped wedge design with $Q$ sequences and $J=Q+1$ time points. Consider the mixed model specified in (\ref{eq_it}) with $\Gamma_j = \mu + \beta_j$ and $C_{ijk} = \alpha_i$, where $\alpha_1,...,\alpha_I \overset{iid}{\sim} N(0,\tau^2)$. Also assume that, conditional on random effects, individual outcomes are independent and normally distributed with variance $\sigma^2$, and define $\phi \equiv \tau^2/(\tau^2+\sigma^2/n)$. From standard results in mixed model theory, we know that the weighted least squares estimator of the fixed effects $(\mu,\beta_2,...,\beta_J,\delta)$ is given by the following, where $X_q$ is the design matrix for an individual assigned to sequence $q$, $V$ is the covariance matrix for a single cluster, and $\bar{Y}_q \equiv (\bar{Y}_q1,...,\bar{Y}_qJ)'$ with $\bar{Y}_qj$ representing the mean responses among individuals in sequence $q$ at time $j$:

\begin{equation*}
	(\hat{\mu},\hat{\beta}_2,...,\hat{\beta}_J,\hat{\delta})' =
	\left( \sum_{q=1}^Q X_q' V^{-1} X_q \right)^{-1}
	\sum_{q=1}^Q X_q' V^{-1} \bar{Y}_q
\end{equation*}\\

Using the form of the covariance matrix $V \equiv \sigma^2((1-\phi)I_J+\phi \mathbb{1}_J)$ implied by the chosen mixed model, where $I_J$ represents a $J \times J$ identity matrix and $\mathbb{1}_J$ represents a $J \times J$ matrix of all ones, calculations done in Matthews and Forbes\cite{matthews2017stepped} show that the treatment effect estimator $\hat{\delta}$ can be expressed as follows:

\begin{equation*}
	\hat{\delta} = \frac{
      Q \sum_{q=1}^Q D_{1q}' V^{-1} Y_q - C_1' V^{-1} \bar{Y}_C
	}{
	  Q \sum_{q=1}^Q D_{1q}' V^{-1} D_{1q} - C_1' V^{-1} C_1
	}
\end{equation*}\\

Above, $\bar{Y}_C \equiv (\sum_{q=1}^Q \bar{Y}_{q1},...,\sum_{q=1}^Q \bar{Y}_{qJ})'$, $D_{1q}$ is the treatment allocation vector for sequence $q$, and $C_1$ is the vector $(0,1,...,Q)'$.\\

We proceed by analyzing each of the four terms $\sum_{q=1}^Q D_{1q}' V^{-1} Y_q$, $C_1' V^{-1} \bar{Y}_C$, $\sum_{q=1}^Q D_{1q}' V^{-1} D_{1q}$, and $C_1' V^{-1} C_1$ separately.\\

We start with the first equality and note the following, where we use the fact that $V^{-1}=\frac{1}{1-\phi} I_J - \frac{\phi}{(1-\phi)(1+J\phi-\phi)}\mathbb{1}_J$ and where $(0,...,0,1,...,1)$ represents a vector of $q$ zeros followed by $J-q$ ones:

\begin{equation*}
\begin{gathered}
	D_{1q}' V^{-1} \bar{Y}_q
	= \left[ \frac{1}{1-\phi}(0,...,0,1,...,1) - \frac{\phi(Q-q+1)}{(1-\phi)(1+J\phi - \phi)}(1,...,1,1,...,1) \right]
	(\bar{Y}_{q1},...,\bar{Y}_{qJ})' \\
	= \frac{1}{1-\phi} \sum_{j=1}^J \left[ I\{j>q\} - \frac{\phi(Q-q+1)}{1+\phi Q} \right]\bar{Y}_{qj}
\end{gathered}
\end{equation*}\\

Above (and throughout), we use the simple fact that $Q+1=J$. The equation above allows us to write the following:

\begin{equation*}
	\sum_{q=1}^Q D_{1q}' V^{-1} \bar{Y}_q
	= \frac{1}{1-\phi} \sum_{q=1}^Q \sum_{j=1}^J \left[ I\{j>q\} - \frac{\phi(Q-q+1)}{1+\phi Q} \right] \bar{Y}_{qj}
\end{equation*}\\

Next, we consider the second term:

\begin{equation*}
\begin{gathered}
	C_1' V^{-1} \bar{Y}_C
	= (0,1,...,Q) \left(\frac{1}{1-\phi} I_J - \frac{\phi}{(1-\phi)(1+J\phi-\phi)}\mathbb{1}_J \right)\bar{Y}_C \\
	= \left[ \frac{1}{1-\phi}(0,1,...,Q) - \frac{\phi Q(Q+1)}{2(1-Q)(1+J\phi - \phi)}(1,...,1) \right] \bar{Y}_C \\
	= \frac{1}{1-\phi} \left( 0-\frac{\phi Q(Q+1)}{2(1+Q\phi)}, 1-\frac{\phi Q(Q+1)}{2(1+Q\phi)}, ..., Q-\frac{\phi Q(Q+1)}{2(1+Q\phi)} \right) \left( \sum_{q=1}^Q \bar{Y}_{q1}, ..., \sum_{q=1}^Q \bar{Y}_{qJ} \right)' \\
	= \frac{1}{1-\phi} \sum_{q=1}^Q \sum_{j=1}^J \left[ j-1-\frac{\phi Q(Q+1)}{2(1+Q\phi)} \right] \bar{Y}_{qj}
\end{gathered}
\end{equation*}\\

We proceed with the third equality and note the following:

\begin{equation*}
\begin{gathered}
	D_{1q} V^{-1} D_{1q}
	= \left[ \frac{1}{1-\phi}(0,...,0,1,...,1) - \frac{\phi(Q-q+1)}{(1-\phi)(1+J\phi-\phi)}(1,...,1,1,...,1)  \right] (0,...,0,1,...,1)' \\
	= \frac{J-q}{1-\phi} \left[ 1-\frac{\phi(Q-q+1)}{1+J\phi-\phi} \right] \\
	= \frac{(Q-q+1)(1+q\phi-\phi)}{(1-\phi)(1+Q\phi)}
\end{gathered}
\end{equation*}\\

This allows us to express the sum as follows, where we use the identities $\sum_{i=1}^n i = \frac{n(n+1)}{2}$ and $\sum_{i=1}^n i^2 = \frac{n(n+1)(2n+1)}{6}$:

\begin{equation*}
\begin{gathered}
	\sum_{q=1}^Q D_{1q} V^{-1} D_{1q}
	= \frac{1}{(1-\phi)(1+Q\phi)} \sum_{q=1}^Q (Q-q+1)(1+q\phi-\phi) \\
	= \frac{1}{(1-\phi)(1+Q\phi)} \left[ \sum_{q=1}^Q (1-\phi(1+Q)+Q) + ((Q+2)\phi-1)\sum_{q=1}^Q q - \phi \sum_{q=1}^Q q^2 \right] \\
	= \frac{1}{(1-\phi)(1+Q\phi)} \left[ Q(1-\phi(1+Q)+Q) + ((Q+2)\phi-1) \frac{Q(Q+1)}{2} - \phi\frac{Q(Q+1)(2Q+1)}{6} \right] \\
	= \frac{Q(Q+1)}{(1-\phi)(1+Q\phi)} \left[ \frac{1}{2} + \frac{\phi(Q-1)}{6} \right]
\end{gathered}
\end{equation*}\\

Finally, the fourth term can be expressed as follows, where we again use the summation identities that we used to analyze the third term:

\begin{equation*}
\begin{gathered}
	C_1' V^{-1} C_1 =
	\frac{1}{1-\phi} \left( 0-\frac{\phi Q(Q+1)}{2(1+Q\phi)}, 1-\frac{\phi Q(Q+1)}{2(1+Q\phi)}, ..., Q-\frac{\phi Q(Q+1)}{2(1+Q\phi)} \right) (0,1,...,Q)' \\
	= \frac{1}{1-\phi} \sum_{q=0}^Q q \left[ q-\frac{\phi Q(Q+1)}{2(1+Q\phi)} \right] \\
	= \frac{1}{1-\phi} \left[ \sum_{q=0}^Q q^2 - \frac{\phi Q(Q+1)}{2(1+Q\phi)} \sum_{q=0}^Q q  \right] \\
	= \frac{Q(Q+1)}{1-\phi} \left[ \frac{\phi Q(Q-1) + 4Q + 2}{12(1+Q\phi)} \right]
\end{gathered}
\end{equation*}\\

Putting these four equalities together, we have the following result:

\begin{equation*}
	\hat{\delta} = \frac{12(1+\phi Q)}{Q(Q+1)(\phi Q^2+2Q-\phi Q-2)}
	\sum_{q=1}^Q \sum_{j=1}^J \left[ Q I\{j>q\} -j+1 +\frac{\phi Q(2q-Q-1)}{2(1+\phi Q)} \right] \bar{Y}_{qj}
\end{equation*}\\

This is the first result of our theorem. To prove the second result, we start by writing the following:

\begin{equation*}
	E\left[\hat{\delta}\right] = \frac{12(1+\phi Q)}{Q(Q+1)(\phi Q^2+2Q-\phi Q-2)}
	\sum_{q=1}^Q \sum_{j=1}^J \left[ Q I\{j>q\} -j+1 +\frac{\phi Q(2q-Q-1)}{2(1+\phi Q)} \right] E\left[ \bar{Y}_{qj} \right]
\end{equation*}\\

Plugging in $E\left[ \bar{Y}_{qj} \right] = \mu + \beta_j + I\{j>q\} \delta_{j-q}$ (with $\beta_1\ equiv 0$) and setting $K \equiv \frac{12(1+\phi Q)}{Q(Q+1)(\phi Q^2+2Q-\phi Q-2)}$, the above is equal to:

\begin{equation*}
\begin{gathered}
	K \sum_{q=1}^Q \sum_{j=1}^J \left[ Q I\{j>q\} -j+1 +\frac{\phi Q(2q-Q-1)}{2(1+\phi Q)} \right] (\mu + \beta_j + I\{j>q\} \delta_{j-q}) \\
	= K\mu \sum_{q=1}^Q \sum_{j=1}^J \left[ Q I\{j>q\} -j+1 +\frac{\phi Q(2q-Q-1)}{2(1+\phi Q)} \right] + \\
	K \sum_{q=1}^Q \sum_{j=1}^J \left[ Q I\{j>q\} -j+1 +\frac{\phi Q(2q-Q-1)}{2(1+\phi Q)} \right] \beta_j + \\
	K \sum_{q=1}^Q \sum_{j=1}^J \left[ Q I\{j>q\} -j+1 +\frac{\phi Q(2q-Q-1)}{2(1+\phi Q)} \right] I\{j>q\} \delta_{j-q}
\end{gathered}
\end{equation*}\\

Next, we show that each of the first two terms is equal to zero. We start with the first term and write:

\begin{equation*}
\begin{gathered}
	K\mu \sum_{q=1}^Q \sum_{j=1}^J \left[ Q I\{j>q\} -j+1 +\frac{\phi Q(2q-Q-1)}{2(1+\phi Q)} \right] \\
	= K\mu \left[ \sum_{q=1}^Q \sum_{j=1}^J QI\{j>q\} - \sum_{q=1}^Q \sum_{j=1}^J j + QJ + \sum_{q=1}^Q \sum_{j=1}^J \frac{\phi Q(2q-Q-1)}{2(1+\phi Q)} \right] \\
	= K\mu \left[ \frac{Q^2(Q+1)}{2} - \frac{Q(Q+1)(Q+2)}{2} + Q(Q+1) + \frac{\phi Q(Q+1)}{2(1+\phi Q)}(0) \right] \\
	=0
\end{gathered}
\end{equation*}\\

Next, we show that the second term equals zero:

\begin{equation*}
\begin{gathered}
	K \sum_{q=1}^Q \sum_{j=1}^J \left[ Q I\{j>q\} -j+1 +\frac{\phi Q(2q-Q-1)}{2(1+\phi Q)} \right] \beta_j \\
	= K \left[ Q \sum_{q=1}^Q (I\{2>q\}\beta_2+...+I\{J>q\}\beta_J) + Q(\beta_2+...+\beta_J) - Q\sum_{j=1}^J j\beta_j \right] \\
	= K \left[ Q \sum_{q=1}^Q (I\{2>q\}\beta_2+...+I\{J>q\}\beta_J) + Q(\beta_2+...+\beta_J) - Q\sum_{j=1}^J j\beta_j \right] \\
	= KQ \left[ (1\beta_2+2\beta_3+...+Q\beta_J) + (\beta_2+...+\beta_J) - (2\beta_2+...+J\beta_J) \right] \\
	= 0
\end{gathered}
\end{equation*}\\

Thus, the first two terms each equal zero, and so we can write:

\begin{equation*}
\begin{gathered}
	E\left[\hat{\delta}\right] = K \sum_{q=1}^Q \sum_{j=1}^J \left[ Q I\{j>q\} -j+1 +\frac{\phi Q(2q-Q-1)}{2(1+\phi Q)} \right] I\{j>q\} \delta_{j-q} \\
	= K \sum_{q=1}^Q \sum_{j=1}^J \left[ Q+1-\frac{\phi Q(Q+1)}{2(1+\phi Q)} - j + \frac{\phi qQ}{1+\phi Q} \right] I\{j>q\} \delta_{j-q} \\
	= K \left[ \frac{(2+\phi Q)(Q+1)}{2(1+\phi Q)} \sum_{q=1}^Q \sum_{j=1}^J I\{j>q\}\delta_{j-q} - \sum_{q=1}^Q \sum_{j=1}^J jI\{j>q\}\delta_{j-q} + \frac{\phi Q}{1+\phi Q} \sum_{q=1}^Q \sum_{j=1}^J qI\{j>q\} \delta_{j-q} \right]
\end{gathered}
\end{equation*}\\

Next, we simplify the three sums in the equation above as follows:

\begin{equation*}
\begin{gathered}
	\sum_{q=1}^Q \sum_{j=1}^J I\{j>q\}\delta_{j-q} = \sum_{j=1}^Q (Q-j+1)\delta_j \\
	\sum_{q=1}^Q \sum_{j=1}^J jI\{j>q\}\delta_{j-q} = \sum_{j=1}^Q \left[ \frac{(Q+1)(Q+2) - j(j+1)}{2} \right] \delta_j \\
	\sum_{q=1}^Q \sum_{j=1}^J qI\{j>q\} \delta_{j-q} = \sum_{j=1}^Q \frac{(Q-j+1)(Q-j+2)}{2} \delta_j
\end{gathered}
\end{equation*}\\

Plugging these identities into the equation above, we obtain:

\begin{equation*}
\begin{gathered}
	E\left[\hat{\delta}\right] = K \left[ \frac{(2+\phi Q)(Q+1)}{2(1+\phi Q)} \sum_{j=1}^Q (Q-j+1)\delta_j - \sum_{j=1}^Q \left[ \frac{(Q+1)(Q+2) - j(j+1)}{2} \right] \delta_j + \frac{\phi Q}{1+\phi Q} \sum_{j=1}^Q \frac{(Q-j+1)(Q-j+2)}{2} \delta_j \right] \\
	= \frac{K}{2(1+\phi Q)} \sum_{j=1}^Q \left[ (j-Q-1)(j+2\phi Qj-Q(1+\phi+\phi Q)) \right] \delta_j \\
	= \sum_{j=1}^Q \frac{6(j-Q-1)((1+2\phi Q)j-(1+\phi+\phi Q)Q)}{Q(Q+1)(\phi Q^2+2Q-\phi Q-2)} \delta_j
\end{gathered}
\end{equation*}\\

This completes the proof.\\

\bibliographystyle{unsrt}
\bibliography{Stepped_wedge_lag}%

\end{document}